\documentclass{article}
\usepackage[utf8]{inputenc}
\usepackage{hyperref}

\usepackage{amsmath,amsthm,amsfonts,amssymb}
\usepackage[T1,T2A]{fontenc}
\usepackage[english]{babel}

\usepackage{amsmath}
\usepackage{amsthm}
\usepackage{amssymb}
\usepackage{amsfonts}

\usepackage{hyperref}
\hypersetup{colorlinks=true,linkcolor=blue,citecolor=blue}
\usepackage{dsfont} 

\usepackage{xcolor}
\usepackage{tikz}
\usepackage{graphics}
\usepackage{graphicx}
\usepackage{color}

\newcommand{\R}{\mathbb{R}}

\newcommand{{\dd}}{\,\mathrm{d}}

\newcommand{\opt}{\mathrm{opt}}
\newcommand{\mes}{\mathrm{mes}}
\newcommand{\sing}{\mathrm{sing}}
\renewcommand{\div}{\mathrm{div}}

\DeclareMathOperator*{\argmax}{arg\,max}

\newtheorem{theorem}{Theorem}
\newtheorem*{theorem*}{Theorem}

\newtheorem{lemma}{Lemma}
\newtheorem*{lemma*}{Lemma}

\newtheorem{proposition}{Proposition}

\newtheorem{remark}{Remark}
\newtheorem{example}{Example}

\title{
Auctions and mass transportation}

\author{ Alexander V. Kolesnikov}

\begin{document}

\maketitle

\begin{abstract}
In this survey paper we present classical and recent results relating the auction design  and the optimal transportation theory. In particular, we discuss in details the seminal result of Daskalakis, Deckelbaum and Tzamos \cite{DDT} about duality between auction design with $1$ bidder and the weak transportation problem. Later investigations revealed the connection of multi-bidder case to the Beckmann's transportation problem. In this paper we overview a number of works on related subjects (monopolist's problem, regularity issues, weak transportation, measure ordering etc.).
In addition, we prove some new results on duality for unreduced mechanisms. 
\end{abstract}

\section{Introduction}

The  optimal transportation theory attracts nowadays substantial attention of researchers in economical science. Among  other applications let us mention the problems of matching, equilibrium,
mechanism design, multidimensional screening, urban planning  and financial mathematics.
Detailed expositions are given in books Galichon \cite{Gal}, Santambrogio \cite{Santambrogio}, but now they are far from being complete, because the number of related articles is rapidly growing every year.

The auction theory is a modern branch in economics, several researchers in this field were awarded with  Nobel Prize in the 21st century.
The aim of this expository paper is to present some classical and new result from the auction theory to a mathematically-minded reader, unfamiliar with economical applications. As the reader will realize, a deep understanding of the auction design model requires a solid mathematical background, which include linear programming and duality theory, PDE's, variational calculus and functional analysis. 
The generally recognized mathematical approach to auctions was given in the celebrated paper of Myerson \cite{Myerson} from 1981, where he completely solved the problem for the case of one good.  But still very little is known in the general case and the author hopes that this short survey could be useful for mathematicians  interested in economical applications.

The first part  of this paper presents the description of the Myerson's model and related classical results. In is followed by a short introduction to the optimal transportation theory. Then we start to discuss recent results. A remarkable relation between optimal transportation and auction theory  for the case of 1 bidder was discovered by Daskalakis, Deckelbaum  and Tzamos in their celebrated paper \cite{DDT}. We present the result from \cite{DDT} and consequent developments and extensions  for many bidders  obtained in Kolesnikov--Sandomirskiy--Tsyvinski--Zimin 
 \cite{KSTZ}, then discuss some open problems and the perspectives. The last section contains new results on characterization of mechanisms with multiple bidders and items by duality methods and relation between  the general auctioneer's problem and the  optimal transshipment problem. 

The author thanks Ayrat Rezbaev and Konstantin Afonin for their help in preparation of this manuscript and Fedor Sandomirskiy for comments and interest.
Author  acknowledges the support of  RSF Grant \textnumero  22-21-00566  https://rscf.ru/project/22-21-00566/. The article was prepared within the framework of the HSE University Basic Research Program.

\section{Auctions}

\subsection{The model: $1$ bidder}

The standard Bayesian auction design deals with the following model:
we consider a set of  $m$ bidders and $n$ items of goods, $n, m\in \mathbb{N}$. Items are supposed to be divisible and normalized in such a way that the total amount of every item is exactly one unit.

 We attribute to every bidder its "private information" vector $x_j \in X = [0,1]^n$, which specifies the willingness of the bidder $j$ to pay for each item. Thus, the utility for bidder $j$ receiving a bundle $p_j \in \mathbb{R}^n_{+}$ of items for the price $t_j \in \mathbb{R}$ can be computed as follows:
 $$
\langle p_j, x_j \rangle - t_j.
 $$

Let us start with the case of $m=1$ bidder.
The private information is distributed according to a given probability law
$$
\rho(x) dx.
$$
The auctioneer offers to the bidder 
\begin{itemize}
    \item Allocation function
    $$
P \colon X \to X,
    $$
    $P = (P_1, P_2, \cdots, P_n)$. Here
    $P_i(x)$ is the amount of $i$-th item which auctioneer sells to the bidder of type $x$.
    \item The  price function $T \colon X \to \mathbb{R}_+$. This is the amount of money the bidder pays for the bundle $P$.
\end{itemize}

 The map $(P,T)$ is called mechanism. 
 The bidder is supposed to report his type $x$ to the auctioneer. 
 In order to prevent the situation when the bidder reports 
 the false value, the auctioneer has to make some restrictions on the mechanism. More precisely, if the the bidder claims to have type $x'$ instead of $x$, the bidder's utility must decrease:
 \begin{equation}
 \label{IC1bid}
\langle x, P(x) \rangle - T(x) \ge \langle x, P(x') \rangle - T(x').
 \end{equation}
 Thus the mechanism must satisfy assumption (\ref{IC1bid}), which is called "incentive compatibility assumption".

 Another natural assumption is assumption of "individual rationality": the utility has to be nonnegative:
 \begin{equation}
 \label{IR1bid}
\langle x, P(x) \rangle - T(x) \ge 0.
\end{equation}
 
The aim of the auctioneer is to maximize the total expected revenue:
 $$
 \int_{X} P \rho dx \to \max
 $$
 over all mechanisms $(P,T)$ satisfying assumptions (\ref{IC1bid}), (\ref{IR1bid}).

\subsection{The model: many bidders}

Let us consider the case of $m>1$ bidders.
The types of bidders are supposed to be distributed according to a probability law $\rho(x_j) dx_j$ and, in addition, we assume that they are chosen independently.

Thus we work with the model space 
$$
\overline{X} = X^m
$$
and use the following notations:
$$
x \in \overline{X}, \ x = (x_1, \cdots,x_m), \  \ \ \ \ \ x_j \in \mathbb{R}^n, \ 1 \le j \le m
$$
$$
x_j = (x_{1,j}, x_{2,j}, \cdots, x_{n,j}), \ \ \ \ \ x_{i,j} \in [0,1], \ \ 1 \le i \le n, 1 \le j \le m.
$$
The space $\overline{X}$ is equipped with the probability measure
$$
\overline{\rho} dx =\bigotimes_{j=1}^m \rho(x_j) dx_j.
$$
When we talk about distribution of a function $f \colon \overline{X} \to \mathbb{R}$, we mean the distribution of $f$ considered as a random variable on the space $(\overline{X},\overline{\rho} dx )$ equipped with the standard Borel sigma-algebra.

 The auctioneer creates $m$ mechanisms $(P_j,T_j)$, where $P_j$ is the bundle of items received by the bidder $j$ for the price $T_j$. Formally, a mechanism (auction) is a map
$$(P,T) \colon X^m \to \mathbb{R}_{+}^{n \times m} \times\mathbb{R}^m, \  
(P,T) = (P_j, T_j)_{1 \le j \le m}.
$$
$$
P_j = (P_{1,j}, \cdots, P_{n,j}) \in \mathbb{R}_{+}^{n}, \ T_j \in \mathbb{R}.
$$

In case of many bidders  we have to add additional assumption of feasibility: 
the mechanism is feasible if for every item $i$
\begin{equation}
\label{feasible}
\sum_{j=1}^m P_{i,j}(x) \le 1.
\end{equation}
This restriction means that the auctioneer has at most one unit of every item to sell.

As before, the auctioneer is looking for the maximum of the expected revenue:
\begin{equation}
\label{AProblem}
\int_{\overline{X}} \bigl( \sum_{1 \le j \le m} T_j  \bigr) \overline{\rho} dx  \to \max.
\end{equation}

As in the case of  one bidder the auctioneer has to solve the problem of preventing the misreport. Informally, this can be done in the following way: the expected revenue of individual bidder under condition that the other bidders report their true types to the auctioneer must decrease if the bidder misreports his type.

Formally, we introduce the following  {\bf marginals of the mechanism}:

\begin{equation}
\label{expP}
\overline{P}_j(x_j) = \int_{X^{m-1}} P_j(x_1, \cdots, x_m) \prod_{i \ne j}  \ \rho(x_i) dx_1 \cdots dx_{j-1} dx_{j+1} \cdots dx_m.
\end{equation}

\begin{equation}
\label{expT}
\overline{T}_j(x_j) = \int_{X^{m-1}} T_j(x_1, \cdots, x_m) \prod_{i \ne j}  \ \rho(x_i) dx_1 \cdots dx_{j-1} dx_{j+1} \cdots dx_m.
\end{equation}

We note that the mapping $(\overline{P}_j, \overline{T}_j)$ (which is called {\bf reduced mechanism}) is nothing else but the conditional expectation of $(P_j,T_j)$ with respect to the random vector $x_j$. We write
$$
(\overline{P}_j, \overline{T}_j)
= \mathbb{E}^{\overline{\rho}} \bigl( (P_j,T_j) | x_j\bigr).
$$

The mechanism is called 
{\bf incentive-compatible} if
	\begin{equation}\label{eq_SP}
	\big\langle\overline{P}_j(x_j),\,x_j\big\rangle-\overline{T}_j(x_j)\geq \langle\overline{P}_j(x_j'),\,x_j\rangle-\overline{T}_j(x_j')
	\end{equation}
 and 
{\bf individually-rational} if
		\begin{equation}\label{eq_IR}
	\big\langle\overline{P}_j(x_j),\,x_j\big\rangle-\overline{T}_j(x_j)\geq 0
	\end{equation}
 for all $1 \le j \le m$.
Thus $(P,T)$ is incentive-compatible/individually rational if and only if all the conditional expectations  $(\overline{P}_j, \overline{T}_j)$ are incentive-compatible/individually rational.

Finally, we are ready to formulate the general  {\bf auctioneer's problem}: maximize the expected revenue
$$
\int_{\overline{X}} \bigl( \sum_{1 \le j \le m} T_j \bigr) \overline{\rho} dx
$$
over all feasible (\ref{feasible}), incentive-compatible (\ref{eq_SP}) and individually-rational (\ref{eq_IR})
mechanisms.

\begin{remark}
    We remark that all the assumptions we made on $T$ use only the marginals $\overline{T}_j$. Thus an optimal mechanism $(P,T)$ does not necessary satisfy a stronger version of  (\ref{eq_IR}):
\begin{equation}
    \label{expostIR}
\big\langle{P}_j(x),\,x_j\big\rangle-{T}_j(x)\geq 0.
    \end{equation}
But violation of (\ref{eq_IR}) might look counterintuitive. However, it is easy to check that $T$ can be modified  in such a way (for a fixed $P$) that  the marginals remain the same and (\ref{expostIR})  holds. Assumption (\ref{expostIR}) is called "ex-post IR" and  (\ref{eq_IR}) is called "ex-ante IR".
\end{remark}

\subsection{Some closed-form solutions}

The formulation of the auction design problem as stated above goes back  to the celebrated paper of Myerson \cite{Myerson}. In this work Myerson proved equivalence of different types of auctions.
Moreover, he proved that in the case of  $n=1$ item
the auction problem  admits an explicit solution. 

\begin{example}
    Let $n=m=1$, $\rho=1$. Then the auctioneer's problem admits the following solution
    $$
 p = I_{[\frac{1}{2},1]}, \ t = \frac{1}{2}I_{[\frac{1}{2},1]}.
    $$
    This is the so-called { "take-it-or-leave-it" mechanism}. Thus the optimal way to sell one item to a bidder  is to price the item at $1/2$. The expected
revenue is $1/4$.
\end{example}
\begin{example}
    { Let $n=1$, $m>1$, $\rho=1$.
    The optimal mechanism is to sell  nothing if all the bidders have types $x_i \le \frac{1}{2}$.
    If it is not the case, sell one item to the bidder with largest $x_i$ at the price which equals to $\max_{j \ne i} \max(x_j, 1/2)$.}
    \end{example}
We discuss in subsections \ref{cfs1}, \ref{mabp} how the solution of Myerson can be obtained by duality methods.

Apart from the $n=1$ case  explicit solutions are rare. The following example was given in \cite{MV}. 
\begin{example}
\label{pavlov-ex}
Let $n=2$, $m=1$ and $\rho$ be the Lebesgue measure on $X=[0,1]^2$.
The square is divided in four parts as shown on Figure~\ref{fig:u_map} (the picture is taken from \cite{KSTZ}).
If the type of the bidder belongs to $\mathcal{Z}$, then the bidder receives no goods and pays nothing.
In the region $\mathcal{A}$ the bidder receives the first good and pays $2/3$, in $\mathcal{B}$ the bidder receives the second good and pays $2/3$. Finally, in the region $\mathcal{W}$ the bidder receives both goods and pays $(4-\sqrt{2})/3$.
\end{example}
This examples demonstrates, in particular, that even in apparently simple cases the solution can not be reduced to one dimensional ones (the optimal mechanism neither sells goods together nor separately). See in this respect \cite{KSTZ}, D.2.

\begin{figure}[!h]
  \centering
  \includegraphics[width=0.3\textwidth]{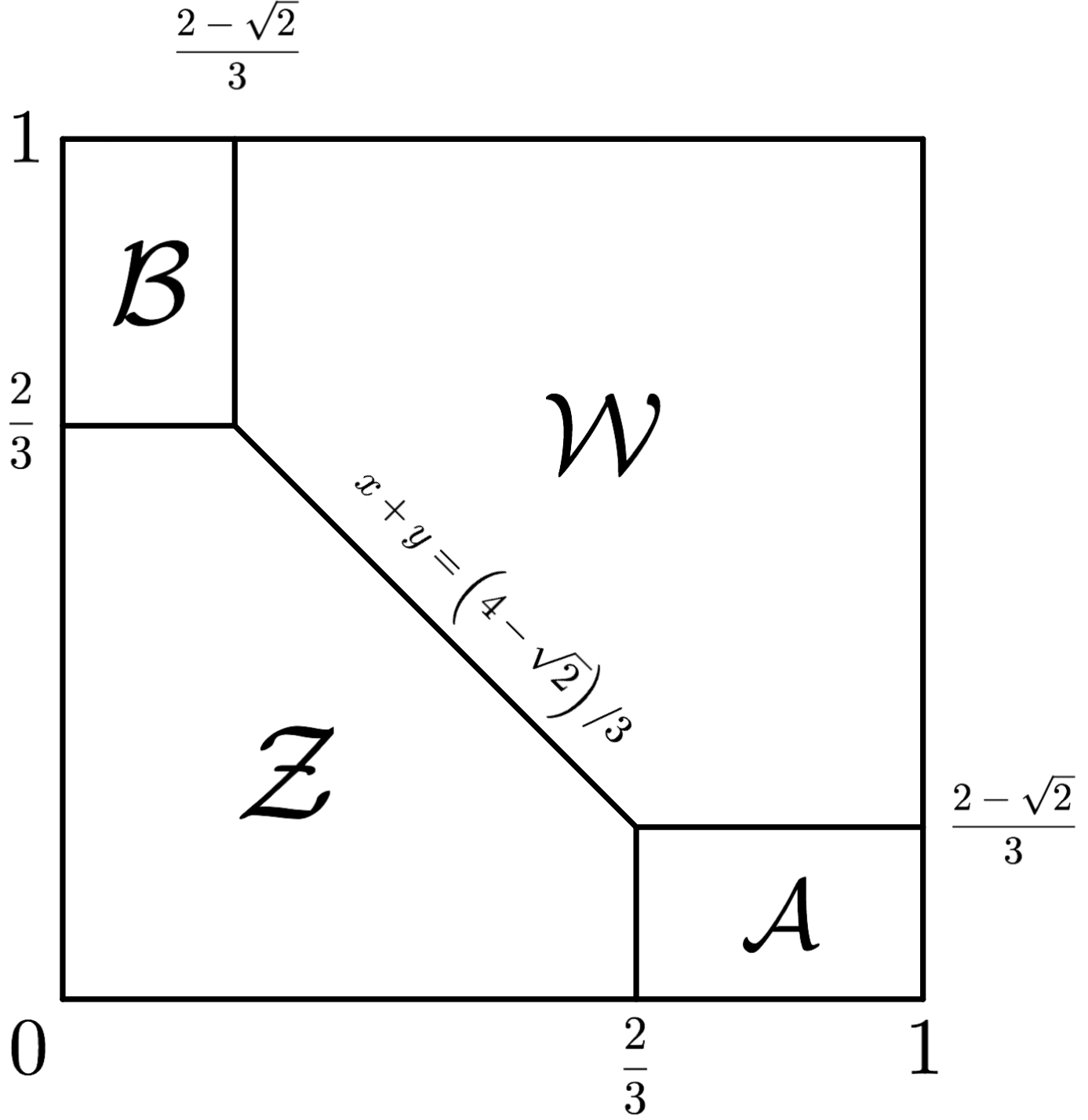}
    \caption{Partition of the square for $n=2, m=1$ and the Lebesgue measure}
   \label{fig:u_map}
\end{figure}

A far-going extension of Example \ref{pavlov-ex} to $m=1, n>2$ was obtained in \cite{GK} by  Giannakopoulos and Koutsoupias. They considered uniform distribution $\rho$ on $X=[0,1]^n$ and conjectured that the so-called  straight-jacket auction (SJA) is optimal in this case. The  straight-jacket auction admits a closed form up to parameters which are roots of certain polynomial.
Giannakopoulos and Koutsoupias computed these roots up to $n=6$ and 
proved that SJA are indeed optimal for $1 \le n \le 6$. Their approach is based on certain duality arguments including some discrete duality results, in particularly, Hall's lemma about matching in bipartite graphs.

Some other examples of explicit solutions can be found in \cite{DDT}.

 There exist 
many results where the authors estimate how some explicit mechanisms are closed to optimal ones.
See, for instance, \cite{HR2017}, \cite{HR2019}, \cite{BILW}, \cite{BGN}.

Unfortunately, in is highly unlikely to find closed-form solutions for $n >1, m>1$. 
 The  result of numerical simulations for $n=m=2$ and the Lebesgue measure of $[0,1]^2$ can be found in \cite{KSTZ}).  
The simulations indicate a complicated structure of the optimal mechanism and suggest that the optimal auction may not admit a closed-form solution even in this benchmark setting.
 

\subsection{The monopolist's problem and  representation of J.-P. Rochet}

In the important particular case of $m=1$ bidder 
the auctioneer's problem can be reduced to a maximization problem for functions (not for mechanisms/mappings). Then the problem turns out to be a particular case of the so-called monopolist's problem.  This approach goes back to Rochet (see \cite{RC} and the references therein).

Let $m=1$. Given a mechanism $(P,T)$ let us consider the {\bf utility function}
$$
u(x) = \langle P(x), x \rangle - T(x)
$$
(note that $(P,T) = (\overline{P}, \overline{T})$ for $m=1$).

Clearly, individual rationality assumption is equivalent to $u \ge 0$. 
The incentive compatibility can be rewritten as follows:
$$
u(x) = \langle P(x), x \rangle - T(x) \ge \langle P(x'), x \rangle - T(x') = u(x') + \langle P(x'), x -x' \rangle .
$$
Then it can be easily concluded that incentive compatibility is equivalent to convexity of $u$ and, in addition, one has:
$$
P = \nabla u
$$
(more precisely, $P \in \partial u$, where $\partial u$ is the subdifferential of $u$).

Thus any incentive compatible and individually-rational mechanism can be recovered from utility $u$.
Clearly, feasibility  $P_{i} \le 1$ and nonnegativity $P_i \ge 0$ is equivalent to condition:
$$
0 \le u_{x_i} \le 1
$$
for all $i$. We say that a function $f$ on $X$ is {\bf increasing}, if for every $i$ function $x_i \to f(x)$ is  increasing  for all fixed values $x_j, j \ne i$. Clearly, for a differentiable function $f$ this is equivalent to assumption $f_{x_i} \ge 0$ for all $i$.

In what follows we denote by 
$
\mathcal{U}
$
the set of {\bf convex increasing functions satisfying
$u(0)=0$}.

Thus we get the following:

\begin{theorem}
\label{RochetRepresentation}{\bf (Rochet) }
For $m=1$ bidder the auctioneer's problem is equivalent to the problem of maximization of
$$
\int_{X} \bigl( \langle x, \nabla u(x) \rangle - u(x) \bigr) \rho(x) dx 
$$
over all convex nonnegative functions $u$ on $[0,1]^n$
satisfying $
0 \le u_{x_i} \le 1
$ for all $i$.

Equivalently maximization can be taken over functions
belonging to $\mathcal{U}$.
\end{theorem}

Let us briefly describe another related problem: the so-called monopolist's problem, going back to Mussa--Rosen \cite{MussaRosen}. Similarly to the auctioneer's problem the initial formulation was given in terms of mechanisms and a reduction  to optimization problem in certain class of utility functions has been obtained in  a seminal paper of Rochet  and Chon{\'e} \cite{RC}.

{\bf Monopolist's problem:} Given a function $\varphi$ and a probability distribution $\rho$ find maximum of
$$
\Phi(u) = \int_{X} \bigl( \langle x, \nabla u(x) \rangle - u(x) - 
\varphi(\nabla u) \bigr) \rho(x) dx 
$$
on the set of convex increasing nonegative functions.

Function $\varphi$ is usually supposed to be convex and is interpreted as a cost of the (multidimensional) product $\nabla u$.
In particular, choosing $\varphi = \delta_{[0,1]^n}$, we get the auctioneer's problem for one bidder.
Here indicator function $\delta_A$ is defined for a given set $A$  as follows : $\delta_A(x) = 0$ for $x \in A$ and $\delta_A(x) = +\infty$ for $x \notin A$.

A systematic study of solutions to the monopolist's 
problem was made in the seminal paper \cite{RC}.
 We emphasize that despite of a classical form of the "energy" functional  $\Phi$ the classical variational approach to the analysis of the monopolist problem does not work here because of  convexity and monotonicity constraints.
In general, any solution $u$ to this problem splits $X$ into three different regions:
\begin{itemize}
    \item Indifference region
    $$
\Omega_0 = \{ u=0 \}
    $$
    \item Bunching region
    $\Omega_1$: $u \ne 0$, $D^2 u$ is degenerated
     \item $\Omega_2$: $u$ is strictly convex: $D^2 u >0$ (nonbunching regions).
\end{itemize}

The word "bunching"   refers to a situation where a group of agents   having different types are treated identically in the optimal solution.
The standard  technique of  calculus of variation is available only inside of $\Omega_2$, where $u$ solves the quasi-linear equation
$$
{\rm div} (\nabla \varphi(\nabla u)) - 1 - {\rm div} (x \cdot \rho)=0.
$$

\subsubsection{Model example and regularity issues}

Rochet and Chon{\'e} suggested a hypothesis about explicit solution for $\varphi(x) = \frac{1}{2}|x|^2$, $\rho=1$ and $n=2$. However, their guess turns out to be wrong. Recently  McCann and Zhang \cite{McCannZhang2} gave a description of the solution partially based on rigorous proofs and partially on numerical simulations. As we will see, despite  the fact that the solution can be described in many details, it looks impossible to give a closed-form solution even in this model case.

The result of McCann and Zhang 
employs regularity property of solutions to monopolist's problem. There exist a limited number of works on regularity 
of solutions to variational problems with convex constraints.
We refer to \cite{CarLac}, \cite{McCannZhang2} and \cite{McCannRankinZhang} and the references therein.
The proofs are delicate and partially based on the advanced regularity technique for the Monge--Ampere equation/optimal transportation. 

In particular, the following theorem (extending earlier result from \cite{CarLac}) was obtained in \cite{McCannRankinZhang} (Theorem 6). Note that the constrains here are slightly different from ours, see,  however, \cite{McCannZhang2}, where these results were applied to the monopolist's problem.

\begin{theorem}
Let $u$ maximize
$$
\int_X \bigl( \langle x, \nabla u \rangle - u - \frac{1}{2} |\nabla u|^2   \bigr) dx
$$
over the set of nonnegative convex functions. Then 
    $u \in C^{1,1}_{loc}(X)$.

    This result is optimal in the following sense: $u \notin C^2(X)$ even for $n=1$.
\end{theorem}

In \cite{McCannZhang2} the following maximization problem was studied:
\begin{equation}\label{R-C}
\int_{[a,a+1]^2} \biggl( \langle x, \nabla v(x) \rangle - v(x) - \frac{1}{2}|\nabla v(x)|^2 \biggr)\,dx\to\max, \quad v\in\mathcal{U}.
\end{equation}

Roughly speaking, there exists four regions: region $\Omega_0$ (triangle), where solution $v$ equals zero; region  $\Omega^0_1$ where $v$ is a function of  $x+y$;
region $\Omega_1^{+} \cup \Omega_1^{-}$, where 
$D^2 v$ has rank one, 
but $v$ has no explicit representation; region $\Omega_2$, where $v$ is strictly convex and solves a linear second-order PDE. The border between  $\Omega_2$ and $\Omega_1^{+} \cup \Omega_1^{-}$ is a "free boundary" and it is a part of the problem to find it. It is  determined by the regularity property $v \in C^{1,1}_{loc}(X)$.

More precisely, the following representation holds:
\begin{eqnarray}\label{in}
\{\Delta v = 0\}&\subset&\Omega_0 = \{(x,y)\in [a,a+1]^2\colon x+y<a+y_0\},\nonumber\\
\{\det D^2 v=0, \ \Delta v > 0\}&\subset&\Omega_1 = \Omega^0_1\cup\Omega^{\pm}_1,\\
\{\det D^2 v>0\}&\subset&\Omega_2=\mathop{\mathrm{cl}}([a,a+1]^2\backslash(\Omega_0\cup\Omega_1)),\nonumber
\end{eqnarray}
where
\begin{eqnarray*}
\Omega^0_1 &=& \{(x,y)\in[a,a+1]^2\colon a+y_0\le x+y\le a+y_1\},\\
\Omega^{\pm}_1&=&\{(x,y)\in\Omega_1\backslash\Omega^0_1\colon \pm(x-y)\ge 0\}.
\end{eqnarray*}

The  solution $v$ has explicit representation on
$\Omega_0$: 
\begin{equation}\label{0}
v(x,y) = 0,\quad (x,y)\in\Omega_0
\end{equation}
 and  on $\Omega_1^0$, where it is degenerated:
\begin{equation}\label{10}
v(x,y) = \frac{3}{8}(x+y)^2-\frac{a}{2}(x+y)-\frac{1}{2}\ln(x+y-2a)+C_0(a),\quad (x,y)\in\Omega^0_1.
\end{equation}
 The shapes of regions $\Omega_1^{\pm}$ are determined by function $h(\theta), R(\theta)$ depending on angle $\theta \in (-\pi/4, \pi/2]$ (see details  in \cite{McCannZhang2}), and function 
$v$ is affine along the segment passing  at angle $\theta$:
\begin{equation}\label{-1}
v(x,y)=v(z(r,\theta)) = m(\theta)r+b(\theta),\quad (x,y)\in\Omega^{-}_1.
\end{equation}
(similarly on $\Omega_1^+$, by symmetry).
Functions $h, R, m,b$ satisfy a system of second-order ODE's, but can  not be computed explicitly. In what follows $n(x,y)$ is the unit normal vector of the corresponding curve.
In the region $\Omega_2$
solve the following free boundary elliptic problem
\begin{equation}\label{2}
\begin{cases}
\Delta v_2(x,y) = 3,&(x,y)\in \mathop{\mathrm{Int}} \Omega_2;\\
\langle\nabla v_2(x,y), n(x,y)\rangle = \langle(x,y), n(x,y)\rangle,&(x,y)\in\partial\Omega_2\cap\partial X;\\
v_2 = \left.v\right|_{\Omega_2},&(x,y)\in\partial\Omega_1\cap\partial\Omega_2
\end{cases}
\end{equation}
Finally, functions $h,R$ must be chosen in such a way that the following  Neumann boundary condition holds:
\begin{equation}\label{neu}
\langle\nabla v_2(x,y), n(x,y)\rangle = \langle\nabla v_1(x,y), n(x,y)\rangle,\quad (x,y)\in\partial\Omega_1\cap\partial\Omega_2,
\end{equation}
here $v_1:=\left.v\right|_{\Omega_1}$.
We need (\ref{neu}) because we know that $v \in C^1(X)$,

The main result of \cite{McCannZhang2} is the following characterization. 

\begin{theorem}
Assume that $v\in\mathcal{U}$ satisfies  \eqref{0}\,--\,\eqref{neu},
and $v$, $\Omega_2$ are Lipschitz.
Then  $v$ is the unique solution to~\eqref{R-C}.
\end{theorem}

\subsubsection{General monopolist's problem}

In this paper we discuss only the classical case of the monopolist's problem
as it was given above. The general model deals with given distribution $\rho dx$ on the space of customers $X \subset \mathbb{R}^n$  and the space of products $Y \subset \mathbb{R}^n$. In the model we a given a product cost $c(y)$ and a function $b(x,y)$  which is the benefit a consumer $x$ gains from a product $y$. The monopolist has to assign the optimal price function $p(y)$. Given $p$ every customer maximizes utility
$$
u(x) = \sup_{y} \bigl( b(x,y) - p(y) \bigr).
$$
Assume that mapping $y(x)$ satisfies $y(x) \in \argmax  \bigl( b(x,y) - p(y) \bigr)$. 
Thus the monopolist is solving  the following optimization problem
$$
\int \Bigl( p(y(x)) - c(y) \Bigr) \rho(y) dy \to \max.
$$
In the classical setting we have $b = \langle x, y \rangle$. The problem can be restated in terms of mechanisms satisfying suitable generalizations of   IC/IR assumptions. In particular, IR assumption turns out to be equivalent to the so-called $b$-convexity. This notion arises in the optimal transport problem with general 
cost function $b$.

Systematical study of this problem, in particular, sufficient condition for convexity of the problem and various regularity issues was made in  \cite{fkmc}, \cite{McCannRankinZhang}, \cite{McCannZhang}, \cite{McCannZhang2}. The analysis is partially based on ideas from the optimal transport theory/ Monge--Amp\'ere equation on Riemannian manifolds.

\subsection{Stochastic domination and reduction to a maximization problem for functions}

Does there exist a reduction in the spirit of Theorem \ref{RochetRepresentation} for many bidders?
The answer is affirmative and the approach heavily relies on the symmetry of mechanism. Let us start with the result of Hart and Reny \cite{HR} for the case of $n=1$ item. Some earlier results in this spirit were obtained by
Matthews \cite{Matthews} and Border  \cite{Border}.

It turns out that the marginals of a feasible mechanism 
can be characterized in terms of stochastic dominance.

Let us consider first $n=1$. 
First, without loss of generality we can restrict ourselves to symmetric mechanisms. If $(P,T)$ is not symmetric, we can symmetrize it : 
$$
 \bigl( P_j^{sym}(x_1, \cdots, x_m), T_j^{sym}(x_1, \cdots, x_m) \bigr)= \frac{1}{m!} \sum_{\sigma} \bigl( P_{\sigma(j)}(x_{\sigma(1)}, \cdots, x_{\sigma(m)}), T_{\sigma(j)}(x_{\sigma(1)}, \cdots, x_{\sigma(m)})\bigr),
$$
where the sum is taken over all the permutations of the bidders.

Clearly, the symmetrized mechanism satisfies the same restrictions and gives the same value to the revenue.
At this step we use that all the bidders are independently and identically distributed.
Thus without loss of generality everywhere below we will talk about {\bf symmetric} mechanisms.

Consider
$\overline{P}$ defined by (\ref{expP}) and for a fixed $\alpha \in \mathbb{R}_+$ set
$
A = \{x :\overline{P}(x) > \alpha\} \subset X.
$
Using symmetry and feasibility of the mechanism one obtains:
\begin{align*}
m \int_{X} \overline{P}(x) I_{A}(x) \rho dx  
& = \int_{\overline{X}} \Bigl( \sum_{j=1}^m P_j (x_j) I_A(x_j) \Bigr) \overline{\rho}dx
\\& \le \int_{\overline{X}}  I_{\exists j: x_j \in A}(x_1, \cdots, x_m) \overline{\rho}dx
= 1 - \Bigl( \int_{A^c} \rho(x) dx\Bigr)^m.
\end{align*}
Put $p = \int_{A^c} \rho dx$. One has
\begin{align*}
    \int_X & \bigl( \overline{P}-\alpha \bigr)_+\rho dx = \int_{A} \overline{P} dx - \alpha(1-p)
    \le \frac{1}{m} - \frac{p^m}{m} - \alpha(1-p)
   \\& \le \frac{1}{m} - \alpha + \frac{m-1}{m} \alpha^{\frac{m}{m-1}} = \int_0^1 (t^{m-1}-\alpha)_{+} dt.
\end{align*}
Since the convex hull of functions $(x-\alpha)_+$ 
is precisely the set of  convex increasing functions, we obtain the following important property:
$$
\int_{0}^1 f(\overline{P}) \rho dx \le \int_0^1 f(t^{m-1}) dx
$$
for every convex increasing function $f$, i.e. the distribution of $\overline{P}$ is {\bf stochastically dominated by} the distribution of $\xi^{m-1}$, where $\xi$ is uniformly distributed on $[0,1]$.
In what follows we write:
$$
\overline{P} \preceq \xi^{m-1}.
$$

It was shown by Hart and Reny that the converse is also true: if $\nu \preceq \xi^{m-1}$, then $\nu$ is a distribution of $\overline{P}$ for some symmetric mechanism $P$.

The general result for many bidders and items unifying 
results of Rochet and Hart--Reny was obtained in \cite{KSTZ}. 
In the same way as for $n=1$ we define
utility function
$$
u(x) = \langle \overline{P}_1(x), x \rangle - \overline{T}_1(x).
$$
for a feasible symmetric mechanism $(P,T)$, satisfying (\ref{eq_IR}) and (\ref{eq_SP}).

Then we prove that $u$ is convex, nonnegative and increasing, see \cite{KSTZ}, Lemma 2. Moreover we conclude that 
$$
 \overline{P}_1 = \nabla u(x)
$$
and, repeating the above arguments, we prove that the distribution of 
$ \overline{P}_{1,i}$ is stochastically dominated by the distribution of $\xi^{m-1}$. Conversely, using the Hart--Reny theorem we can show that every $u \in \mathcal{U}$ satisfying $u_{x_i} \preceq \xi^{m-1}$ defines feasible, individually rational, incentive compatible mechanism on $X^m$ with the same total revenue (\cite{KSTZ}, Lemma 3).
Thus we get the following:

\begin{theorem}
\label{mainreductiotheorem}
    The value of the general auctioneer's problem
    coincides with the maximum of
    $$
m \int_{X} \bigl( \langle x, \nabla u(x) \rangle - u(x) \bigr) \rho(x) dx 
$$
over all convex nonnegative increasing functions $u$ on $[0,1]^n$
satisfying $
u_{x_i} \preceq  \xi^{m-1}
$  for all $0 \le i \le n$, where $\xi$ is uniformly distributed on $[0,1]$.
\end{theorem}

\begin{remark}
    It is easy to check that without loss of generality one can assume that $u(0)=0$. Sometimes it is convenient to write functional in the form
     $$
m \int_{X} \bigl( \langle x, \nabla u(x) \rangle - u(x) \bigr) \rho(x) dx + m u(0),
$$
which can be applied to functions without restriction $u(0)=0$.
\end{remark}

\section{Mass transportation problem}

In this section we briefly recall  some facts about optimal transportation theory, which are important for application to the auction design. A comprehensive presentation of the mathematical theory of optimal transportation can be found in  \cite{Santambrogio}, \cite{Villani}, \cite{BoKo} and  \cite{BoKo2023} (in Russian).

\subsection{General optimal transportation problem}

Let $\mu$ and $\nu$  be probability measures on measurable spaces $X$ and $Y$, and let $c: X \times Y \to \mathbb{R}$ be a measurable function.
The classical Kantorovich  problem is the minimization problem  
$$
\int_{X \times Y}c(x, y) d\pi \to \inf
$$
on the space  $\Pi(\mu, \nu)$  of probability measures on $X \times Y$ with fixed marginals  $\pi_x = \mu$ and $\pi_y = \nu$.

It is well-known that this problem is closely related to another linear programming problem, which is called 
``dual transportation problem'' 
$$
\int_X f d\mu + \int_Y g d\nu \to \sup.
$$
The dual transportation problem is considered on the couples of  functions $(f(x), g(y))$, with $f \in L^1(\mu), g \in L^1(\nu)$, satisfying $f(x) + g(y) \le c(x, y)$ for all $x \in X$, $y \in Y$.

According to the "easy part" of the duality theorem
$$
\int_X f d\mu + \int_Y g d\nu 
\le \int_{X \times Y}c(x, y) d\pi 
$$
for all feasible $\pi, (f,g)$.
These problems are linear programming problems and 
 under broad assumptions the values of both problems coincide.
Complementary slackness condition
gives that solutions $\pi$ and $(f,g)$ satisfy
$$
f(x) + g(y) = c(x,y)
$$
for $\pi$-almost all $(x,y)$.

We say that $T$ is the optimal transportation mapping or solution to the Monge problem if there exists a Kantorovich solution $\pi$, such that $\pi$ is concentrated on the graph of $T$:
$$
\pi\{ (x,T(x)), x \in X\} =1.
$$
It is well-known that Monge solutions exist for a wide class of cost functions, including functions  $c(x,y) = |x-y|^p$, where $p \ge 1$ and $| \cdot |$ is Euclidean norm and $X=Y=\mathbb{R}^N$. If $T$ is a solution of the Monge problem, then $\nu$ is the image of $\mu$ under $T$: $\nu =\mu \circ T^{-1}$.
In addition, for every solution $(f,g)$ of the dual problem one has
\begin{equation}
\label{monge-fgT}
f(x) + g(T(x)) = c(x,T(x)).
\end{equation}
 
\subsection{$L^1$-transportation problem}

The case when $X=Y = \mathbb{R}^N$ and  $c(x,y) = |x-y|$, where $| \cdot |$ is  arbitrary norm, is special. In this case without loss of generality one assume, in addition, that $f=-g$ for any dual solution $(f,g)$. Moreover, $f$ is a $1$-Lipschitz function with respect to $|\cdot |$:
$|f(x) - f(y) | \le |x-y|$.
Thus the value of the dual problem equals
$$
W_1(\mu-\nu) = \max\left\{\int\limits_X ud(\mu-\nu):  |\nabla u|_* \le 1 \right\},
$$
where $|\cdot|_*$ is the dual norm. 

The functional $W_1$ is a norm on the space of the signed finite measures.

To understand informally the structure of the solution let us rewrite the problem in the form
$$
\int_X u d (\mu - \nu) - \int_X \delta_{|\cdot|_*}(\nabla u) dx \to \max, 
$$
where  

$$
\delta_{|\cdot|_*}(x) = 
\begin{cases}
0,& |x|_* \le 1\\
+\infty,& |x|_*>1
\end{cases}
$$
The formal Euler-Lagrange equation of the above functional looks as following
$$
\mu-\nu = - {\rm div} (\nabla \delta_{|\cdot|_*}(\nabla u)).
$$
Note that $\partial\delta(x)=0$ if $|x|_*<1$ and
$$\partial\delta(x)= \sigma \cdot N(x) , \sigma \in [0,+\infty),$$ if $|x|_*=1$,
where $N(x) = \nabla (\frac{1}{2} |x|^2_*) $
(in particular, $N(x)=x$ for Euclidean the norm).

Hence one can rewrite the above equation in the form
\begin{equation}
    \label{l1densityequation}
\mu-\nu = - {\rm div} (\sigma \cdot N(\nabla u)),
\end{equation}
where $\sigma\ge 0$ is unknown function called transportation density.

To understand better the meaning of $u, \sigma$, we observe that by (\ref{monge-fgT}) the optimal transport mapping $T$ pushing forward $\mu$  onto $\nu$ satisfies
$$
u(x) - u(T(x)) = |x-T(x)|.
$$
Since $u$ is $1$-Lipschitz, this means that $u$ is affine on $[x,T(x)]$.

Further analysis  gives the following picture: the space $X$ can be split  $X = \cup_{\alpha} I_{\alpha}$ (up to a set of $\mu$-measure zero) into union of segments $I_{\alpha}$, which are called  transport rays. Transportation of mass goes along these rays. The potential $u$ is affine on every ray, thus $u$ is responsible for the directions of transportation. The function $\sigma$ is responsible for the balance of mass under transportation.  

$L^1$-transportation  problem turned out to be quite non-trivial technically, it was studied in \cite{Sud}, \cite{AP1}, \cite{EGan}, \cite{CFMC}. For results on transportation density and relation to other type of transport (Beckmann, congested transport...) see \cite{Santambrogio} and the references therein.

 More generally,   given a cost $c(x,y)$ and a signed measure $\pi$ on $X=Y$ satisfying $\pi(X)=0$, one can consider optimal  {\bf transshipment problem}
 \begin{equation}
     \label{OTP}
\int_{X \times Y} c(x,y) d \gamma \to \min,
 \end{equation}
 where $\gamma$ satisfies $\gamma_X - \gamma_Y = \pi$. The dual problem takes the form
 $$
\int_X u d \pi \to \max
 $$
 over functions satisfying $u(x) - u(y) \le c(x,y)$.

\subsection{Beckmann's problem}

Beckmann's problem was introduced in the seminal  paper \cite{Beck}
to model commodity flows. Mathematically it can be stated as follows: we are given a (convex) function $\varphi$  and a measure $\pi$ on $\mathbb{R}^N$
 with $\pi(\mathbb{R}^N)=0$. Then 
$$
{\rm Beck}_{\varphi}(\pi) = \inf\left\{\int_{\mathbb{R}^N} \varphi(c(x)) dx, \quad -\operatorname{div} (c)=\pi \right\}, 
$$
where 
$c$ is a vector field and  divergence is understood in the weak sense
$$
\int_{\mathbb{R}^N} \left\langle \nabla f,c \right\rangle dx= 
\int_{\mathbb{R}^N} fd \pi,
$$
for all sufficiently smooth $f$ with compact support.

If $\varphi(x)=|x|$ is a norm on $\mathbb{R}^N$, the Beckmann's problem is equivalent to  $L^1$-transportation problem. The detailed explanation can be found in \cite{Santambrogio}. Informally, let us note that if $u$ is a solution to $L^1$-transportation problem and $c$ satisfies $-\operatorname{div} (c)=\mu-\nu$,  one has
$$ 
W_1(\mu-\nu) = \int u d(\mu - \nu) =
- \int u \operatorname{div} (c)
= \int \langle \nabla u, c \rangle dx
\le \int |c|  dx.
$$
Hence $ {\rm Beck}_{|\cdot|}(\pi) \ge  W_1(\mu-\nu)$.
In the other hand, if one takes $c = \sigma N(\nabla u)$ from equation (\ref{l1densityequation}), one gets
$$
W_1(\mu-\nu) = \int \langle \nabla u, c \rangle dx
=\int |c|  dx.
$$
This implies the claim.

Finally, let us mention that
Beckmann's problem is equivalent to a Monge-Kantorovich-type problem called ``congested optimal transport'';  
see \cite{Santambrogio} for the detailed presentation and references. 
This problem is an optimization problem on the set of curves, it admits a dynamical formulation,
in spirit of the Benamou--Brenier formula.

\subsection{Weak transport}

Let  $\mu, \nu$ be probability measures on $X$.
Assume we are given a cost function
$$
c:X\times \mathcal{P}(X)\to \mathbb{R}_+,
$$
where $\mathcal{P}(X)$ is the space of probability measures on $X$.
Consider the following nonlinear optimization problem:
\begin{equation}\label{eq1}
K_c(\nu|\mu):=\inf_{\sigma \in \Pi(\mu,\nu)}\int_X c(x,\sigma^x)\mu(dx),
\end{equation}
where $\sigma^x$ are conditional distributions, i.e. identity
$$
\int_{X\times X}f(x,y)\sigma(dx\,dy)=
\int_X\int_X f(x,y)\sigma^x(dy)\mu(dx)
$$ holds
for every bounded Borel $f$.
This problem is called nonlinear mass transportation problem. We refer to \cite{BPR} and the references therein.

An important particular case of this problem is given by 
$c(x,y) = |x - \int y \sigma^x(dy)|$ (more generally: $c(x,y) = |x - \int y \sigma^x(dy)|^p$), where $|\cdot|$ is a norm. 
This is the so-called  weak mass transportation transport. It appeared in in { \cite{Marton}}, as a tool for studying concentration inequalities. The systematic study of this problem started in \cite{Goz}.

An important property of the weak transport is the relation to {\bf convex domination} of probability measures. We write
$$
\mu \preceq_c \nu
$$
($\mu$ is dominated by $\nu$ )
if and only if 
$$
\int f d \mu \le \int f d \nu,
$$
where $f$ is arbitrary convex function. Then the following result holds (see \cite{ACJ}, \cite{BVP}) :
\begin{equation}
\label{weak-dualweak}
\inf_{\sigma \in \Pi(\mu,\nu)} \int \bigl|x - \int y \sigma^x(dy)\bigr| \mu(dx)
= \inf_{\eta \preceq_c \nu} W_1(\mu,\eta).
\end{equation}

\begin{remark}
    Recall that according to the Strassen theorem
    relation $\mu_1 \preceq_c \mu_2 \preceq_c \cdots \preceq_c \mu_m$
    is equivalent to existence of a martingale $\xi_1, \xi_2, \cdots, \xi_m$ satisfying
    $$
\xi_i \sim \mu_i, \ 1 \le i \le m.
    $$ 
     A similar result holds for submartingales with stochastic domination $\preceq$ instead of convex domination $\preceq_c$.
\end{remark}

The following duality result for  the weak transportation  problem  has been obtained in \cite{Goz}:
\begin{equation}
\label{weakdualL1}
\inf_{\sigma \in \Pi(\mu,\nu)} \int \bigl|x - \int y \sigma^x(dy)\bigr| \mu(dx) = \sup\left\{\int_X\varphi d\mu-\int_X\varphi d\nu\right\},
\end{equation}
where the supremum is taken over all convex $1$-Lipschitz functions. This will be the key observation for applications of mass transport to auctions.
As we will see in the next section, the auctioneer's problem for $1$-bidder is  equivalent to the dual weak transportation problem (up to a slightly different domain of feasible functions).

The weak transportation  problem attracts 
a lot of attention  withing the last decade. We refer  to the nice survey paper \cite{BVP}, where the reader can find a long list of  applications
(measure concentration, Schr\"odinger problem,  martingale optimal transport). 
Let us mention, in particular,  a nice proof by duality of the Strassen theorem for two marginals (see \cite{Goz}).
Apart from auctions, some applications of weak transport in economics  (labor theory) can be found in \cite{Cho}.

\section{Mass transportation in $1$-bidder case}

A remarkable relation of the $1$-bidder  problem  to the mass transportation problem was discovered by 
Daskalakis, Deckelbaum  and Tzamos in \cite{DDT}.
The main idea of \cite{DDT} was to rewrite the auctioneer's functional
using integration by parts. Indeed,  assume that $\rho$ is sufficiently regular (at least $C^1(X)$). Then
 \begin{align*}
&  m \int_{X} \bigl( \langle x, \nabla u(x) \rangle - u(x) \bigr) \rho(x) dx + m u(0)
= \\& m \Big( \int_{\partial X} u \langle x, \nu \rangle
\rho d \mathcal{H}^{n-1} - \int_{X} \bigl[ \langle x, \nabla \rho(x) \rangle + (n+1) \rho(x) \bigr] u dx + u(0) \Bigr). 
\end{align*}
Here $\mathcal{H}^{n-1}$ is the Hausdorff measure of dimension $n-1$ and $\nu$ is the outer normal to $\partial X$.
We rewrite the auctioneer's problem in the form
$$
\int_X u d\mu \to \max, 
$$
where $u$ is convex and increasing and
$$
\mu = m \Bigl(  \langle x, \nu \rangle \rho \mathcal{H}^{n-1}
- (n+1) \rho dx -  \langle x, \nabla \rho(x) \rangle dx + \delta_0 \Bigr)
$$
is a signed measure with zero balance: $\mu(X)=0$.
We will call $\mu$ {\bf transform measure}. In what follows
$$
\mu = \mu_+ - \mu_{-}
$$
will be the decomposition of $\mu$ into mutually singular
positive and negative parts.

Finally, we get the following reformulation of the {\bf auctioneer's problem with $1$ bidder:}
find
$$
\max_{u \in \mathcal{U}, u \in {\rm Lip}_1(X)} \int u d \mu 
$$
on the set of convex increasing $1$-Lipschitz functions. 

In this form the problem looks quite similar to the dual transportation problem, especially to the dual weak $L^1$-transportation problem. Thus, the main result of 
Daskalakis, Deckelbaum  and Tzamos can be viewed as a variant of relations (\ref{weak-dualweak}), (\ref{weakdualL1}) for a {\bf stochastic} domination instead of convex.
In what follows the notation $\| \cdot \|_1$ means the standard $l^1$-norm:
$$
\| x\|_1 = \sum_{i=1}^n |x_i|.
$$

\begin{theorem}\label{DDT1bid}
    {(\bf Daskalakis--Deckelbaum--Tzamos, duality for $1$ bidder)}. The $1$-bidder problem is equivalent to a weak transportation problem. In particular, the following duality relation holds:
    $$
\max_{u \in \mathcal{U}, u \in {\rm Lip}_1(X)} \int u d \mu 
= \min_{\gamma \in \Gamma} \int_{X^2}\|x-y\|_{1} \gamma(dx dy),
    $$
    where $\Gamma$ is the set of nonnegative measures on $X^2$
    satisfying $Pr_1 \gamma \succeq \mu_+$, $Pr_2 \gamma = \mu_{-}$.
\end{theorem}
\begin{remark}
    Moreover, another duality relation holds for the related weak transshipment problem:
    $$
\max_{u \in \mathcal{U}, u \in {\rm Lip}_1(X)} \int u d \mu 
= \min_{\gamma : \pi \succeq \mu} \int_{X^2}\|x-y\|_{1} \gamma(dx dy),
    $$
    where $\pi = Pr_1 \gamma - Pr_2 \gamma$. 
    \end{remark}
The proof of Theorem \ref{DDT1bid} is based on the application of the so-called {\bf Fenchel--Rockafellar duality}. Given two convex functionals $\Theta, \Xi$ with values in $\mathbb{R}\cup \{+\infty\}$ on a normed space $E$
one has the following  form of the minimax principle:
\begin{equation}
\label{FR}
\inf_{z \in E} \bigl( \Theta(z) + \Xi(z)\bigr) = \max_{z^* \in E^*}
\bigl(  -\Theta(-z^*) - \Xi(z^*) \bigr),
\end{equation}
provided there exists a point $z_0$ such that $\Theta(z_0)<\infty, \Xi(z_0) < \infty$, and $\Theta$ is continuous at $z_0$. The latter conditions is needed for application of the Hahn-Banach separation theorem and establishing the existence of the maximum point at the right-hand side of the duality.  The proof can be found in \cite{Villani}.

The direct application of the  Fenchel--Rockafellar duality to the functionals 
$$
\Theta (f) = 
 \begin{cases}
      0,   & \mbox{if} \ f(x,y) \ge  - \|x-y\|_1 \\
      + \infty, &  \mbox{in the opposite case}
  \end{cases}
$$
$$
\Xi (f) = 
 \begin{cases}
      \int \psi d\mu_{-} - \int \phi d\mu_{+}, &  \mbox{if}  f(x,y) = \psi(y) - \phi(x), \mbox{where} \ \psi, \phi \in \mathcal{U}(X) \\
      + \infty, &  \mbox{in the opposite case.}
  \end{cases}
$$
gives
$$
\min_{Pr_1 \gamma \succeq \mu_+,  Pr_2 \gamma \preceq \mu_-} \int_{X^2}\|x-y\|_{1} \gamma(dx dy)
= \sup_{\phi,\psi \in \mathcal{U}, \phi(x) - \psi(y) \le \|x-y\|_1 }
\Bigl( \int \phi d\mu_+ - \int \psi d\mu_-  \Bigr),
$$
where $\gamma$ is a measure on $X^2$ with
$$
\gamma_+(X^2) = \gamma_{-}(X^2) = \mu_+(X) = \mu_{-}(X).
$$

The proof will be completed if we show that  without loss of generality one can assume $\varphi=\psi$ and replace condition $Pr_1 \gamma \succeq \mu_{+}$
with $Pr_1 \gamma = \mu_{+}$.

The first statement was proved in \cite{DDT}, Lemma 5. Given $\phi \in \mathcal{U}$ define $\overline{\psi}(y) = \sup_x (\varphi(x) - \|x-y\|_1)$. It is clear that
$$
\int \phi d\mu_+ - \int \psi d\mu_- 
\le \int \phi d\mu_+ - \int \overline{\psi} d\mu_{-}. 
$$
Then  one can show that $\overline{\psi} \in \mathcal{U}$ and, moreover, $\overline{\psi}=\phi$. Thus the replacement of $(\phi,\psi)$ with $(\phi, \overline{\psi}) = (\phi, \phi)$ will increase the value of the functional.

The second statement was proved in \cite{BVP} Theorem 6.1. by the following argument: given $\gamma$ with $Pr_1 \gamma \succeq \mu_+,  Pr_2 \gamma \preceq \mu_{-}$
let $\xi, \eta$ be $X$-valued random variables satisfying
$$
\xi_1 \sim \mu_+, \\ \eta_2 \sim \mu_-, \\ ( \xi_2, \eta_1) \sim \gamma.
$$
and
$$
\xi_1 \le \mathbb{E}(\xi_2|\xi_1), \ 
\eta_1 \le \mathbb{E}(\eta_2|\eta_1).
$$
Such variables can be easily built with the help  of  the Strassen theorem.
We will prove the statement if we show that 
there exists a variable $Z$ satisfying
$Z \preceq \mu_{-}$ and 
$
\mathbb{E} \|\xi_2 - \eta_1\|_1 \ge \mathbb{E} \|\xi_1  - Z\|_1.
$
Define
$$
Z = \xi_1 + \mathbb{E}(\eta_1 - \xi_2|\xi_1).
$$
Indeed,
\begin{align*}
    Z \preceq Z + \mathbb{E}(\xi_2-\xi_1|\xi_1)
    = \mathbb{E}(\eta_1| \xi_1) \preceq \eta_1 \preceq \eta_2 \sim  \mu_{-}.
\end{align*}
Finally, by  Jenssen's inequality
$$
\mathbb{E} \|\xi_2 - \eta_1\|_1 \ge 
\mathbb{E} \| \mathbb{E}(\xi_2 - \eta_1 |\xi_1)\|_1 =
\mathbb{E} \|\xi_1  - Z\|_1.
$$

\begin{example} (Example \ref{pavlov-ex} revisited).
Consider the problem from Example \ref{pavlov-ex}.
The solution $u^{opt}$ and the majorizing measure $\pi^{opt}$ are given below (see the details in \cite{DDT}, in particularly, description of the optimal transportation mapping). 

The optimal function $u^\opt$ is given by
\begin{equation}
\label{uauctionsol}
    u^\opt(x,y) = \begin{cases}
    0 & (x,y) \in  \mathcal{Z} \\
    x - \frac{2}{3}& (x,y) \in \mathcal{A} \\
    y -\frac{2}{3} & (x,y) \in \mathcal{B} \\
    x+y - \frac{4-\sqrt{2}}{3} & (x,y) \in \mathcal{W}.
    \end{cases}
\end{equation}

The transform measure $\mu$  and the optimal ``imbalance'' measure $\pi^\opt$ majorizing  $\mu$ look as follows: 
\begin{equation}\label{eq_m_uniform}
\mu =   \delta_0 
+ \lambda_1|_{[0,1] \times \{0\}} + \lambda_1|_{\{0\} \times[0,1] }
- 3 \lambda_2|_{[0,1]^2},
\end{equation}
\begin{equation}\label{eq_pi_uniform}
\pi^\opt =    \lambda_1|_{[0,1] \times \{0\}} + \lambda_1|_{\{0\} \times[0,1] }
-3 \lambda_2|_{[0,1]^2 \setminus \mathcal{Z}},
\end{equation}
where $\lambda_2, \lambda_1$ are the two- and one-dimensional Lebesgue measures, respectively.
\end{example}

\begin{remark}
    An alternative proof of Theorem \ref{DDT1bid} was proposed in \cite{KM}.
\end{remark}

\section{Duality for the case of $m>1$ bidders}

In this section we discuss a duality statement for $m>1$. The results were obtained in \cite{KSTZ}. We show, in particular, that the dual auctioneer's problem is   naturally related to the Beckmann's transportation problem.

\subsection{Duality in the monopolist's problem}

An important auxiliary result is the duality theorem in monopolist's problem. 
Informally, the duality relation looks as follows: given a convex function 
$\varphi$ and a measure $\mu$ of finite variation satisfying $\mu(X)=0$ one has:
$$
\sup_{u \in \mathcal{U}} \Bigl( \int_{X} ud\mu - \int_{X} \varphi(\nabla u) \rho dx \Bigr)  = \inf_{c: -{\rm div} (c \cdot \rho) \succeq \mu}\int_X \varphi^*(c) \rho dx. 
$$
Here $\mathcal{U}$ is the set of convex increasing functions on $X$ and divergence is a signed measure ${\rm div} (c \cdot \rho)$. We always understand  it in the weak (integration by parts) sense:
for a smooth function $\psi$ on $X$ one has
$$
\int_X  \psi d{\rm div} (c \cdot \rho) = -  \int_X \langle \nabla \psi, c \rangle \rho dx.
$$
As usual, inequality $\sup \le \inf$ is easy to verify:
\begin{align*}
 \int_{X} ud\mu - \int_{X} \varphi(\nabla u) \rho dx 
&  \le -\int_{X} u d{\rm div} (c \cdot \rho) - \int_{X} \varphi(\nabla u) \rho dx 
 = \int_X \langle \nabla u, c \rangle \rho dx - \int_{X} \varphi(\nabla u) \rho dx 
 \\& \le \int_X \varphi^*(c) \rho dx. 
\end{align*}
By complementary slackness argument we conclude that in the absence of  duality gap the extremizers $u_0,c_0$ must satisfy
$$
c_0 \in \partial \varphi(\nabla u_0).
$$

The rigorous proofs of different versions of this result can be found 

1. In \cite{KSTZ} (Proposition 6.7). The proof relies on a version of minimax theorem (Sion's theorem), where one of the underlying spaces is supposed to be compact. This leads to essential restrictions on $\varphi$.

2. In \cite{McCannZhang2}. Function $\varphi$ is assumed to be strictly convex. 
In the McCann--Zhang formulation the assumption
$-{\rm \div}(c \cdot \rho) \succeq \mu$ 
is replaced by another assumption, which does not presume any regularity of the vector field $c$:
$$
\int_X u d \mu \le \int_X \langle c, \nabla u \rangle \rho dx,
$$
for all $u \in \mathcal{U}$
(see also Section 4 below).
The proof is based on the Fenchel--Rockafellar duality and some regularization of the problem.

3. 
 Another relatively simple and general proof based on the   Fenchel--Rockafellar duality 
  was given in \cite{TBK}.

To understand informally the duality statement let us observe the following simple relation: given a convex function $\varphi$ and a distribution $\rho$ we define the energy functional:
 \begin{equation}
     \label{energy}
\Theta (u) = 
\begin{cases}
        \int_X  \varphi(\nabla u)  \rho dx, \ & u \in W^{1,1}(X),  \\
      + \infty, &  \mbox{in the opposite case}.
   \end{cases}
  \end{equation}
  Here the standard notation $W^{1,1}(X)$ is used for the classical Sobolev space.
  \begin{theorem}
  \label{beck=leg-energy}
    The Legendre transform  $$
\Theta^*(\mu) = \sup_{u} \Bigl(\int u d \mu - \Theta(u) \Bigr).
$$ of the energy functional 
(\ref{energy}) 
 coincides with  the Beckmann's transportation functional (\ref{beckmann-functional}) for the dual potential $\varphi^*$.
\begin{equation}
     \label{beckmann-functional}
\Theta^*(\mu) = 
\begin{cases}
   + \infty , &  \mbox{if} \ \mu(X) \ne 0,
\\ \inf_{ c : {\rm div}  (\rho \cdot c) = - \mu } 
\int_X \varphi^*(c) \rho dx, & \mbox{if} \ \mu(X)=0.
\end{cases}
\end{equation}
\end{theorem}

As usual, it is easy to see that one functional dominates the other: indeed, for all admissible  $(u,\mu, c)$ one has
$$
\int u d \mu - \Theta(u)  = - \int u d{\rm div} (\rho \cdot c) - \Theta(u)
= \int \langle \nabla u, c \rangle \rho dx - \int \varphi(\nabla u) \rho dx
\le \int \varphi^*(c) \rho dx.
$$
Hence $$ \inf_{ c : {\rm div}  (\rho \cdot c) = - \mu } 
\int_X \varphi^*(c) \rho dx \ge \sup_{u} \Bigl(\int u d \mu - \Theta(u) \Bigr). 
$$
The proof of the equality $\inf=\sup$ can be found (in slightly different settings) in \cite{BP}, \cite{Santambrogio}, \cite{TBK}.

Below we give a sketch of the proof of duality Theorem  \ref{dual=beck} for monopolist's problem from  \cite{TBK}.

   {\bf Corrigendum.}
     {  \it We warn the reader that it was mistakenly claimed  in \cite{TBK} that 
 the infimum in the duality relation can be always reached on some measure $\pi$ of finite variation (see in this respect Subsection \ref{opp}). However, the statement holds true in the form $\max=\inf$.
 }

\begin{theorem}
\label{dual=beck}
    Let $\varphi$ be a convex lower semicontinuous function, finite on $X =[0,1]^n$ and taking $+\infty$-value outside of  $X$. Let 
    $\mu \in \mathcal{M}_0$, where $\mathcal{M}_0$ is the space of measures with finite variation satisfying $\mu(X)=0$. Then the following duality relation holds:
    $$
    \max_{u \in \mathcal{U}(X)}
    \Phi(u)
    = \inf_{\pi \in \mathcal{M}_0: \mu \preceq \pi} {\rm Beck}_{
\rho, \varphi^*}(\pi), 
    $$
    where 
    \begin{equation}
    \label{phibeck}
\Phi(u) = \biggl( \int_X u d\mu - \int_X \varphi(\nabla u) \rho d x\biggr), \ 
 {\rm Beck}_{\rho, \varphi^*}(\pi)
    = \inf_{c\colon \pi + {\rm div}(c \cdot \rho) =0} \int_X \varphi^*(c) \rho dx.
    \end{equation}
   \end{theorem}

 {\bf Sketch of the proof of Theorem \ref{dual=beck}:} 
 
   We give a proof of this result under assumption that $\varphi$ is bounded on $X$ and infinite on $\mathbb{R}^n \setminus X$. 
   
 We consider functionals $\Theta, \Xi$ on the space of continuous functions  $E = C(X)$ equipped with the uniform norm and   functionals $\Theta^*, \Xi^*$ on the dual space $\mathcal{M}(X)$ of signed measures with finite variation.
 We denote by $E_0 \subset E$ the subspace of functions satisfying $x(0)=0$ and by 
 $\mathcal{M}_0(X)$ the subset of $\mathcal{M}(X)$ satisfying $\mu(X)=0$. 
 
Functional $\Theta$ is given by (\ref{energy}). Let us consider functional
$$
\Xi (u) = 
 \begin{cases}
      - \int_X u d\mu, & u \in \mathcal{U} \\
      + \infty, &  \mbox{in the opposite case}
  \end{cases}
$$
It is easy to check that $\Xi$ is convex and its Legendre transform satisfies
\begin{align*}
\Xi^*(\nu) & = \sup_{u \in \mathcal{U}} \biggl( \int_X u d\mu +\int_X u d\nu \biggr)   = \begin{cases}
  + \infty , & \mbox{if} \ \nu(X) \ne 0,
  \\ \sup_{u \in \mathcal{U}} \int_X u d(\mu + \nu), & \mbox{if} \ \nu(X)=0.
\end{cases}
\end{align*}

Note that
$$
- \max_{u \in \mathcal{U}_0} \biggl( \int_X u d\mu - \int_X \varphi(\nabla u)
\rho dx  \biggr) 
= \min_{ u  \in E} \Bigl( \Theta(u) + \Xi(u)\Bigr). 
$$
We apply the following variant of the Fenchel--Rockafellar duality, which we prove later:
\begin{equation}
\label{mm}
\inf_{ u  \in E} \Bigl( \Theta(u) + \Xi(u)\Bigr)
= -\inf_{\nu \in \mathcal{M}(X)} \Bigl( \Theta^*(-\nu) + \Xi^*(\nu))\Bigr) .
\end{equation}
One has
\begin{align*}
 \max_{u \in \mathcal{U}} &  \biggl( \int_X u d\mu - \int_X \varphi(\nabla u)
\rho dx  \biggr)   = \inf_{\nu \in \mathcal{M}(X)} \Bigl( \Theta^*(-\nu) + \Xi^*(\nu))\Bigr) 
\\& = \inf_{\nu \in \mathcal{M}_0, \nu \preceq - \mu} \Theta^*(-\nu)
= \inf_{\nu \in \mathcal{M}_0, \mu \preceq \nu} \Theta^*(\nu)
= \inf_{\nu \in \mathcal{M}_0, \mu \preceq \nu} {\rm Beck}_{\rho,\varphi^*}(\nu).
\end{align*}

Let us prove (\ref{mm}). Note that (\ref{FR}) is different from (\ref{mm}), because (\ref{mm}) does not claim that the minimum can be reached. We follow the proof of  \cite[Theorem 1.9]{Villani}.
Relation (\ref{mm}) is equivalent to the following:
$$
\inf_{ u  \in E} \Bigl( \Theta(u) + \Xi(u)\Bigr)
= \sup_{\nu \in E^*} \inf_{u,v \in E} \biggl( \Theta(u)  + \Xi(v) + \int_X (u-v) d \nu\biggr).
$$
Taking $u=v$, one gets
$$
\inf_{ u  \in E} \Bigl( \Theta(u) + \Xi(u)\Bigr)
\ge \sup_{\nu \in E^*} \inf_{u,v \in E} \biggl( \Theta(u)  + \Xi(v) + \int_X (u-v) d \nu\biggr).
$$

It is sufficient to prove that there exists a sequence $\nu_n \in \mathcal{M}$ such that
$$
\Theta(u)  + \Xi(v) + \int_X (u-v) d \nu_n
+ \frac{1}{n} \ge    \inf_{ u  \in E} \Bigl( \Theta(u) + \Xi(u)\Bigr)
\quad \forall\, u,v \in E.
$$
Since $\Theta, \Xi$ 
a both invariant with respect to addition of a constant, the desired inequality is possible if and only if  $\nu_n \in \mathcal{M}_0$. Thus from the very beginning we take $u,v \in E_0$. Note that the space  of functionals   $E^*_0$
can be identified with  $\mathcal{M}_0$.

Thus the problem is reduced to the following: prove existence of  $\nu_n \in \mathcal{M}_0$ such that
$$
\Theta(u)  + \Xi(v) + \int_X (u-v) d \nu_n
+ \frac{1}{n} \ge  \inf_{ u  \in E_0} \Bigl( \Theta(u) + \Xi(u)\Bigr) \ \
\forall\,  u,v \in E_0.
$$

Let us define
$$
C_n = \Bigl\{ (u,t) \in E_0 \times \mathbb{R};\ 
\max_{x \in X} \varphi(x) + \frac{1}{n} \ge  t \ge \Theta(u)+\frac{1}{n} \Bigr\},
$$

$$
K = \{ (v,s) \in E_0 \times \mathbb{R}; \ s \le  \inf_{ u  \in E_0} \bigl( \Theta(u) + \Xi(u)\bigr) -  \Xi(v) \}.
$$

Note that $C_n$ and $K$ don't intersect. They are convex, $K$ is closed and $C_n$ is { compact} because $\varphi$ is bounded on $X$, hence the derivatives of $u$ are uniformly bounded for all $(u,t) \in C_n$.
By the Banach-Khan theorem there exists a non-zero separating functional $(\omega_n, \alpha)$, 
$\omega_n \in E^*_0 = \mathcal{M}_0$, $\alpha \in \mathbb{R}$:
$$
\omega_n(u) + \alpha t \ge \omega_n(v) + \alpha s, 
$$
if $ \max_{x \in X} \varphi(x) + \frac{1}{n} \ge  t \ge \Theta(u)+\frac{1}{n}$, $ s \le  \inf_{ u  \in E_0} \bigl( \Theta(u) + \Xi(u)\bigr) -  \Xi(v)$, $u,v \in E_0$.
It is possible only if $\alpha>0$. Set $\nu_n = \omega_n/\alpha$. Dividing inequality by $\alpha$, one gets
$$
\int_X u d \nu_n + \Theta(u) + \frac{1}{n} \ge \int_X v d \nu_n + \inf_{ u  \in E_0} \bigl( \Theta(u) + \Xi(u)\bigr) - \Xi(v).
$$
The proof is complete.

We conclude this subsection with an important auxiliary result from \cite{KSTZ} (Proposition 3), stated here in a more general form.
This is a priori estimate which can be quite useful for studying  maxima of the monopolist's functional.

\begin{proposition}
\label{ae-exist-lip}
Let $\varphi$ be convex non-decreasing and lower-semicontinuous function with $\varphi(0)=0$.
 For every function $u \in \mathcal{U}$, there exists a non-decreasing convex function $\tilde{u}$ with $\tilde{u}(0) = 0$ such that
\begin{equation}
\label{main-ae}
 \langle x, \nabla \tilde{u}(x) \rangle - \tilde{u}(x)  - \varphi(\nabla \tilde{u}) 
 \ge \max\left\{ \langle x, \nabla u(x) \rangle - u(x)  - \varphi(\nabla u),\, 0\right\}
\end{equation}
for all $x \in X$. In particular, this implies that for any function $u^{\mathrm{opt}} \in \mathcal{U}$ maximizing the functional $\Phi(u)$ over $u\in \mathcal{U}$, the inequality
\begin{equation}
\label{mainapest}
 \langle x, \nabla {u^\opt}(x) \rangle - {u^\opt}(x)  - \varphi(\nabla u^{opt})  \ge 0
\end{equation}
holds almost everywhere.
\end{proposition}

\subsection{Beckmann's problem and duality for many bidders}

Recall that the Beckmann's functionnal is given by
$$
 {\rm Beck}_{\rho, \varphi^*}(\pi)
    = \inf_{c\colon \pi + {\rm div}(c \cdot \rho) =0} \int_X \varphi^*(c) \rho dx.
    $$
Several forms of duality for the multibidder case have been established in \cite{KSTZ}.
The formal proof is given by the following line on computations:

\begin{align*}
   &  \max_{u \in \mathcal{U}(X), u_{x_i} \preceq \xi^{m-1}} \int_X \bigl( \langle x, \nabla u(x) \rangle - u(x) \bigr) \rho dx  = \max_{u \in \mathcal{U}(X), u_{x_i} \preceq \xi^{m-1}} \int_X u d \mu \\& = 
     \max_{u \in \mathcal{U}(X)} \min_{\varphi_i \in \mathcal{U}([0,1])} \Bigl( \int_X u d \mu - \int_X \sum_i \varphi_i(u_{x_i})  \rho dx + \sum_i \int_0^1 \varphi_i(t^{m-1}) dt  \Bigr)
     \\& = 
     \min_{\varphi_i \in \mathcal{U}([0,1])}  \max_{u \in \mathcal{U}(X)} \Bigl( \int_X  u d \mu  - \int_X \sum_i \varphi_i(u_{x_i}) \rho dx + \sum_i \int_0^1 \varphi_i(t^{m-1}) dt  \Bigr)
    \\& =
     \inf_{\varphi_i \in \mathcal{U}([0,1]), \mu \preceq \pi} \Bigl( {\rm Beck}_{\rho, \sum_i \varphi^*_i}(\pi) + \sum_i \int_0^1 \varphi_i(t^{m-1}) dt  \Bigr).
\end{align*}
 
Here we formally interchange $\min$ with $\max$ (minimax principle) and use the  duality for the monopolist's problem.
The rigorous proof is, however, tedious and applies various instruments from functional analysis.
We state the final result in the following theorem:

\begin{theorem}\label{th_vector_fields_inf}
    In the auctioneer's problem  with $m \geq 1$ bidders, $ n \geq 1$ items, and bidders' types distributed on $X=[0,1]^{n}$ with positive density $\rho$, the optimal revenue coincides with 
\begin{equation}\label{eq_vector_fields_inf}
m\cdot \inf_{
  \footnotesize{\begin{array}{c}
       (\varphi_{i})_{1 \le i \le n},\\  
       \pi\succeq_{} \mu
  \end{array}}}\left[ \Bigl( {\rm Beck}_{\rho, \sum_i \varphi^*_i}(\pi) + \sum_{i =1}^n \int_0^1\varphi_{i}\left(t^{m-1}\right){\dd} t\right],
 \end{equation}   
where $\Phi$ is given by $\sum_{i=1}^n \varphi_i(x_i)$ and $\varphi_{i}\colon \R_+\to \R_+\cup\{+\infty\}$ are non-decreasing convex functions with $\varphi_{i}(0)=0$ for each item $i$. 
\end{theorem}

Another Theorem from \cite{KSTZ}
establishes a duality relation in the form $\max=\min$.
To this end let us consider  the set $\mathcal{C}^\mes$ of  non-negative vector measures
 $\varsigma=(\varsigma_{i})_{i\in\mathcal{I}}$ satisfying the following condition 
\begin{equation}\label{eq_condition_on_c_mes}
\int_X \bigl(\langle \nabla u(x),\, x \rangle - u(x) \bigr)\cdot \rho(x) {\dd} x \leq \sum_{i=1}^n\int_X \frac{\partial u}{\partial x_{i}}(x){\dd} \varsigma_{i}(x)
\end{equation}
for any $u \in \mathcal{U}$. 
By the Lebesgue decomposition theorem, each  $\varsigma_{i}$ can be represented as the sum of the component that is absolutely continuous with respect to $\rho(x){\dd} x$ and the singular one. We get
\begin{equation}\label{eq_Lebesgue_decomposition}
{\dd}\varsigma_{i} = c_{i}(x)\cdot \rho(x){\dd} x + {\dd}\varsigma_{i}^{\sing}(x).    
\end{equation}
If the singular component is absent and $c=(c_i)_{i\in \mathcal{I}}$ is smooth, we can define 
$\pi=-\div_\rho[c]$ and see that the condition~\eqref{eq_condition_on_c_mes} is equivalent to the familiar majorization condition   $-{\rm div} (c \cdot \rho) \succeq m$.

    It is known that the minimizers of the Beckmann's functional belong to  the space $\mathcal{C}^\mes$, which can viewed as a natural completion of the space of integrable vector fields (see Theorem 4.6 in \cite{Santambrogio}). We observe  the same phenomenon considering the dual auctioneer's problem.

\begin{theorem}[Extended dual]\label{th_vector_fields_min}
The optimal revenue in the auctioneer's problem coincides with 
\begin{equation}\label{eq_vector_fields_min}
m\cdot \min_{
  \footnotesize{\begin{array}{c}
       \varphi_{i} \in \mathcal{U}[0,1],\\  \varsigma \in \mathcal{C}^\mes 
  \end{array}}}
 \sum_{i =1}^n\left(\varsigma_{i}^\sing(X)+\int_X  \varphi^*_{i}\big(c_{i}(x)\big) \rho(x){\dd} x +\int_0^1\varphi_{i}\left(t^{m-1}\right){\dd} t\right)
 \end{equation} 
 and the minimum 
 is attained. 
\end{theorem}

\begin{remark}
It is possible to obtain the closed form-solution to the dual problem in the setting of Example \ref{pavlov-ex} 
    (see \cite{KSTZ}, D). The solution is in general {\bf non-unique}. Moreover, one can construct different types of solutions:  regular vector fields  and extended vector fields from $\mathcal{C}^{mes}$.
\end{remark}

\subsection{Closed-form solution for $1$-item case by duality}
\label{cfs1}

Solution for $1$-utem case was obtained by  Myerson in \cite{Myerson}. We present below a sketch of duality arguments from \cite{KSTZ}, D1 skipping  some technical details. Function $\rho$ is supposed to be positive and continuous.

In what follows $\varsigma$ is a probability measure on $[0,1]$ and 
 $V$ is the so-called virtual valuation function
$$
V(x) = x - \frac{1 - \mathcal{P}(x)}{\rho(x)}, \quad \mathcal{P}(t) = \int_0^t \rho(x)\,{\dd}x.
$$

The approach based on the following characterization, which can be verified by elementary means (one-dimensional integration by parts formula, see D1 in  \cite{KSTZ}),

\begin{proposition}
For any smooth $u$, 
$$
\int_0^1 (x \cdot u'(x) - u(x))\rho(x)\,{\dd}x = \int_0^1 u'(x)\cdot V(x)\rho(x)\,{\dd}x - u(0),
$$

Denote $F_V(t) = \int_t^{1} V(x)\rho(x)\,{\dd}x$, $F_\varsigma(t) = \varsigma\left([t, 1]\right)$. Then the inequality
\begin{equation}\label{eq:1D_majorization}
\int_0^1 (x \cdot u'(x) - u(x))\rho(x)\,{\dd}x \le \int_0^1 u'(x)\,{\dd}\varsigma
\end{equation}
holds for all smooth convex increasing $u$ with $u(0)=0$ if and only if $F_\varsigma(t) \ge F_V(t)$ for all $t \ge 0$.
\end{proposition}

By duality the value of the primal problem equals
\begin{equation}
\label{1dimdual}
\inf_{\varphi, \varsigma} \int_0^1 \varphi^*(c) \rho dx + \int_0^1 \varphi(t^{m-1}) dt,
\end{equation}
where $\varsigma = c \cdot \rho$, $\varphi$ is convex increasing and $\varsigma$ satisfies 
(\ref{eq:1D_majorization}).

Without loss of generality one can consider only increasing $c$ (because optimal $c$ satisfies $c \in \partial \varphi(\partial u)$). 
Rewrite inequality $F_\varsigma(t) \ge F_V(t)$ 
in the form 
$$
F_\varsigma(\mathcal{P}^{-1}) \ge 
F_{V}(\mathcal{P}^{-1}).
$$
Set
$$
g(t)=F_\varsigma(\mathcal{P}^{-1}) = \int_t^1 c(\mathcal{P}^{-1}(x)) dx.$$
Note that $g$ is nonegative, non-increasing and concave (because $c$ is increasing). Moreover, any such function can be represented in this way. Thus we conclude that the dual solution $c$ does not exceed the concave envelope $G$ of 
$
F_{V}(\mathcal{P}^{-1})
$
(minimal non-negative non-increasing concave function pointwise above $F_{V}(\mathcal{P}^{-1})$). One can prove that this envelope has the form
$$
G = F_{\overline{V}}(\mathcal{P}^{-1}), \ \ F_{\overline{V}}(t) = \int_t^{1} \overline{V}(x)\rho(x) dx,
$$
where
$$
\overline{V} = - G'(\mathcal{P})
$$
is the so-called ironed virtual valuation function.

This means that the choice 
$
c = \overline{V}
$
is the best possible.
Thus, the  solutions $u$ and $c = \overline{V}, \varphi$ must satisfy
$$
\int_0^1 u' V \rho dx = 
\int_0^1 u' \overline{V} \rho dx
= \int_0^1 \varphi^*(\overline{V}) \rho dx + \int_0^1 \varphi(t^{m-1}) dt
$$

Choosing the best $\varphi$ in we get by the complementary slackness (see 
\cite{KSTZ}, D1)
$$
\overline{V}(x) \in \partial \varphi(\mathcal{P}^{m-1}(x))
$$
for almost all $x$
and
$$
\int_0^1 \varphi(u')\rho dx = \int_0^1 \varphi(t^{m-1}) dt.
$$
Analysing these results we arrive at the following conclusion:

\begin{example}
    Let $n=1$ and function $V$ be increasing. Then $\overline{V} = \max\{V,0\}$ and  $\{\overline{V}=0\} = [0,x_0]$ for some $x_0 > 0$. Then the function
    $$
u(x)  = \begin{cases}
0 , &  x \in [0,x_0] \\ 
\int_{x_0}^{x} \mathcal{P}^{m-1}(t) dt, & x \in [x_0,1].
\end{cases}
    $$
is a solution to the reduced auctioneer's problem.
\end{example}

The case of increasing $V$ is called regular. 
In the non-regular case the ironed virtual valuation $\overline{V}$ determines segments where $u'$ is constant and does not coincide with  $\mathcal{P}^{m-1}$. This phenomenon is called "ironing"
(see details in \cite{Myerson}, \cite{Dask}).



\section{Optimal mechanism as a solution to a Beckmann-type problem with constraints}

In this section we come back to the initial (unreduced) problem and consider  the space
$$
\overline{X} = X^m = ([0,1]^n)^m.
$$
The main aim of this section is to find  natural formulations of the 
auctioneer's problem on $\overline{X}$ in terms of optimal transportation/ Beckmann's transportation/ dual transportation problem.

\subsection{Beckmann's problem on $\overline{X}$ with linear constraints}

In this subsection we get 
a Beckmann's type formulation of the dual auction design problem on the space $\overline{X}=X^m$.
We note that the duality Theorem \ref{th_vector_fields_inf} can be formulated without functions $\varphi_i$. Indeed, for a fixed $c$ the functions $\varphi_i$  are solution to an auxiliary one-dimensional optimisation problem which has a closed-form solution (this is actually, the optimal transportation problem). Applying this optimization, the reader gets the following result (see \cite{KSTZ} for details):

\begin{proposition}
\label{nophi2}
The value of the auctioneer's problem is
$$
 \inf_{P} \mathbb{E}^P
  \sum_{1 \le i \le n} 
  \max (x_{i,1}, x_{i,2}, \cdots, x_{i,m}), 
$$
where for every fixed  $i$
the  random variables $x_{i,j}, 1 \le j \le m$ are independent and have distribution $Pr_i P$. The infimum is taken over distributions $P$
 on  $\mathbb{R}^{n}$  satisfying the following  property: there exists a vector field  $c \in \mathcal{C}$
  with distribution $Pr_i P$.
\end{proposition}

We equip  $\overline{X}$
 with the product measure $$\overline{\rho}(x) d x = \bigotimes_{j=1}^m \rho({x}_j) d {x}_j$$
 and the   norm
$$  |x|_{\overline{X}} =
\sum_{1 \le i \le n} 
  \max (|x_{i,1}|, |x_{i,2}|, \cdots, |x_{i,m}|).
  $$
In what follows we will also use the dual norm
$$  |x|_{\overline{X}^o} =
\max_{1 \le i \le n} 
  \sum_{j=1}^m |x_{i,j}|.
  $$
For every vector field
$$
c = (c_1, \cdots, c_{n})
$$ on $X$ let us define the extended vector field on $\overline{X}$
\begin{equation}
    \label{cextend}
\overline{c}({x}_1, \cdots, {x}_m) 
= (c({x}_1), \cdots, c({x}_m)).
\end{equation}
We will apply the following notation:
\begin{equation}
\label{cextend-meas}
    \overline{\pi} = \sum_{j=1}^m \rho({x}_1 ) d {x}_1 \times \cdots  \times \pi(d {x}_j) \times \cdots 
\rho({x}_m ) d {x}_m.
\end{equation}
Note that
$$
{\rm div}_{\overline{\rho}} \overline c
= \sum_{j=1}^m \rho({x}_1 ) d {x}_1 \times \cdots  \times {\rm div}_{\rho} c({x}_j) \times \cdots 
\rho({x}_m ) d {x}_m.
$$
 We get immediately the following result  from Proposition \ref{nophi2}: 
\begin{proposition} \label{nophi3} {\bf (Auction's design problem as a Beckmann's problem with linear constraints)}
The auction design problem is equivalent to the following minimization problem: 
$$
\inf_{\pi \succeq \mu, \overline{c} } \int_{\overline{X}} |\overline{c}|_{\overline{X}}  
\overline{\rho}({x}) d {x},
$$
where  $\overline{c}$ satisfies the following assumptions:
\begin{enumerate}
\item
$\overline{c}$ has the form 
\eqref{cextend}
\item
$
{\rm div}_{{\rho}}  c
= - {\pi}.
$
\end{enumerate}
\end{proposition}

\subsection{Optimal mechanism and the Beckmann's problem}
\label{mabp}

In this subsection we prove a duality relation 
on  $\overline{X}$.
Since we restrict ourselves to symmetric mechanisms, all the functions and vector fields are supposed to be symmetric in this section. 
Recall that a function $f : \overline{X} \to  \overline{X}$ is symmetric 
if $f \circ T_{\sigma} = f$ and 
a field  $
F \colon \overline{X} \to \overline{X}$ is  symmetric, if 
$$
F(T_{\sigma}) = T_{\sigma}(F).
$$
Here $\sigma$ is arbitrary permutation of $m$ elements and
$T_\sigma(x_1, \cdots,x_m) = (x_{\sigma(1)}, \cdots, x_{\sigma(m)})$.

Let $L$ be the space of { $L^2(\overline{\rho}dx)$-integrable} vector fields 
$$
F \colon \overline{X} \to \overline{X}, \
F = (F_1, F_2, \cdots, F_m), \ F_{j} = (F_{1j}, \cdots, F_{nj}),
$$ with the following properties:
\begin{itemize}
\item 
$F$ is symmetric
\item 
For every $k$ the conditional expectation of $F_k$ with respect to $x_k$ vanishes
$$
\mathbb{E}^{\overline{\rho}} ( F_{ik}| {x}_k)
= \int_{\prod_{j \ne k} X_j} F_{ik} \prod_{j \ne k} \rho({x}_i) d{x}_j =0, \ 1 \le i \le n, \ 1 \le j \le m.
$$
\end{itemize}

In what follows $L^{\bot}$ denotes the space of symmetric fields orthogonal to $L$ in $L^2(\overline{\rho}dx)$-sense. These are precisely the $L^2(\overline{\rho}dx)$-fields having the form
$$
\overline{c}(x) = (c(x_1), c(x_2), \cdots, c(x_m)).
$$
Finally, we denote by
$$Lip^{sym}(| \cdot|_{\overline{X}}, L)$$ 
 the 
the set of symmetric Lipschitz functions  $v$ with the following property: there exists $F \in L$ such that
$$
| \nabla v + F|_{\overline{X}^\circ} \le 1
$$
$\overline{\rho}dx$-a.e.
The subset 
$$Lip^{sym}_{+}(| \cdot|_{\overline{X}}, L) \subset 
Lip^{sym}(| \cdot|_{\overline{X}}, L)$$
are precisely the functions satisfying additional assumption
$$
\nabla v + F \ge  0, 
$$
i.e. $\partial_{x_{i,j}} v + F_{i,j} \ge 0$ a.e. for all $1 \le i \le n, 1 \le j \le m$.

\begin{theorem}
    The values of the optimization problems below coincides with the value of the auctioneer's problem:
    \begin{enumerate}
        \item 
        $$
m \cdot \max_u \int_{X} u d \mu = \max_v \int_{\overline{X}} v d \overline{\mu},
        $$
        where maximum is taken over  $u \in \mathcal{U}_0$ such that, in addition, 
        $v \in Lip^{sym}_{+}(| \cdot|_{\overline{X}}, L) $, where
        $$v(x) = \sum_{j=1}^m u(x_j)$$
        \item 
        $$
\inf_{{\pi} \succeq {\mu}} \sup_{v  \in Lip^{sym}(| \cdot|_{\overline{X}}, L)} \int_{\overline{X}} v d\overline{\pi}.
        $$
    \end{enumerate}
\end{theorem}
\begin{proof}

Denote by $R$  the value of the auctioneer's problem.

To prove the first statement let us come back to the canonical formulation of the auctioneer's problem in terms of mechanisms. Let $u$, $v$ satisfy the assumption of the theorem, in particular, there exists  $F \in L$ satisfying $|\nabla v + F|_{\overline{X}^0} \le 1$. Let us find the corresponding feasible  mechanism  $(P,T)$,  satisfying  IR, IC assumptions. Set $P = \nabla v + F$. 
Clearly, $P$ is feasible and   $\overline{P}_1 = \mathbb{E}^{\overline{\rho}}(P_1|x_1) = \nabla u(x_1)$. Set $T_1(x_1)  = \langle x_1, \nabla u(x_1) \rangle - u(x_1)$.
Assumptions IR, IC hold, because $u$ is an increasing convex function satisfying $u(0)=0$.
Conversely, we know  (see explanations in subsection 2.5), 
that given a symmetric feasible IR, IC mechanism $(P,T)$, one can define 
$u \in \mathcal{U}_0$ by formula $\overline{P}_1 = \mathbb{E}^{\overline{\rho}}(P_1|x_1) = \nabla u(x_1)$. Then function
$v(x) = \sum_{j=1}^m u(x_j)$ satisfies assumptions of item  1, because  we can take $F = P - \nabla v$.

To prove the second statement let us show first that
$
\inf_{{\pi} \succeq {\mu}} \sup_v \int_{\overline{X}} v d\overline{\pi}
        $ is not bigger than $R$. Indeed, take a feasible $v$ and a vector field $\overline{c} \in L$ satisfying ${\rm div}_{\overline{\rho}} \overline{c} = - \overline{\pi}$. Then
        $$
    \int_{\overline{X}} v d \overline{\pi} = -  
    \int_{\overline{X}} v d ({\rm div}_{\overline{\rho}} \overline{c})
    = \int_{\overline{X}} \langle \nabla v, \overline{c} \rangle \overline{\rho} dx  = \int_{\overline{X}} \langle \nabla v + F, \overline{c} \rangle \overline{\rho} dx \le \int_{\overline{X}} |\overline{c}|_{\overline{X}} \overline{\rho} dx. 
        $$
Hence
$$
 \inf_{{\pi} \succeq {\mu}} \sup_{v} \int_{\overline{X}} v d \overline{\pi}
 \le  \inf_{\pi \succeq \mu, {\overline{c}} } \int_{\overline{X}} |\overline{c}|_{\overline{X}}  
\overline{\rho}({x}) d {x}.
$$
        But according to Proposition \ref{nophi3} the right-hand side is precisely $R$.

        To see that the value of the functional coincides with  $R$, note that  for every ${\pi} \succeq {\mu}$ one can choose   $v = \sum_{j=1}^m u(x_j)$, where $u \in \mathcal{U}_0$. Then
        $$
\int_{\overline{X}} v d \overline{\pi}
 = m \int_{X} u d \pi \ge m \int_X u d\mu
 = \int_{\overline{X}} v d \overline{\mu}.
        $$
        By the previous point the maximum value of the right-hand side on the set of functions $v = \sum_j u(x_j)$ ($u$ is convex increasing) with $ v \in Lip^{sym}_{+}(| \cdot|_{\overline{X}}, L)$  is exactly $R$. Hence
        $$
\sup_{v \in Lip^{sym}(| \cdot|_{\overline{X}}, L)}  \int_{\overline{X}} v d \overline{\pi} \ge R.
        $$ Thus we get $
\inf_{\pi \succeq \mu} \sup_{v \in Lip^{sym}(| \cdot|_{\overline{X}}, L)}  \int_{\overline{X}} v d \overline{\pi} \ge R
        $ and the  proof is complete.
\end{proof}

It is clear from the proof that any optimal mechanism
$P$ admits the structure
\begin{equation}
    \label{punreduced}
P = \nabla v_0 + F, \ F \in L,
\end{equation}
where $v_0 = \sum_{j=1}^m u_0(x_j)$ and $u_0$ is a solution to the reduced problem.

Let us describe some immediate consequences of this representation. For the sake of simplicity we shall assume that there exists an optimal vector field $c$ without singular components.

    \begin{remark}
        We do not know whether this true in general, see Theorem \ref{th_vector_fields_min}.
    \end{remark}

Let us apply the complementary slackness condition:
$$
    \int_{\overline{X}} v d \overline{\pi}  = \int_{\overline{X}} \langle \nabla v + F, \overline{c} \rangle \overline{\rho} dx =\int_{\overline{X}} |\overline{c}|_{\overline{X}} \overline{\rho} dx. 
        $$
        Equality is possible if and only if
        $$
        \langle P, \overline{c} \rangle = 
\langle \nabla v + F, \overline{c} \rangle  = |\nabla v + F|_{\overline{X}^o}
|\overline{c}|_{\overline{X}} = |\overline{c}|_{\overline{X}}.
        $$
Equivalently 
$$
\sum_{i=1}^n \sum_{j=1}^m c_{i,j}(x_j) P_{i,j}(x) =  \sum_{i=1}^n \max_{1 \le j \le m} c_{i,j}(x_j).
$$
This leads to the following observation:
\begin{proposition} Assume that there exist an optimal $c$ without singular components. Then almost all $x \in \overline{X}$ satisfy the following property: if the  collection of $m$ valuation functions
$$
( c_{i,1}(x_1),   \cdots, c_{i,m}(x_m))
$$
has the unique non-zero maximum $c_{i,j_0}(x_{j_0})$, then the optimal mechanism is choosing this maximum: $P_{i,j_0}(x)=1$ and $P_{i,j}(x)=0$ for $j \ne j_0$. This means that the $i$-th item is given to the $j_0$-th bidder.

If the sequence of valuation functions has many non-zero maxima, then the optimal mechanism distributes the item between them.
\end{proposition}

\begin{remark}
Note that this rule can be extended to the case when optimal $c$ has singular components. Indeed, if $c_{i,j_0}$ admits a singular component $\varsigma_{i,j_0}$, then it is concentrated on the set $u_{x_{i,j_0}}(x_{j_0})=1$ (see Corollary 2 in \cite{KSTZ}). If $u_{x_{i,j}}(x_{j})<1$ for other $j$, this means that the $i$-th item should be given  to the $j_0$-th bidder. 
\end{remark}

We finish this subsection with an observation that under assumptions below the virtual valuation $c$ can be replaced by $\nabla u$.
Assume:
\begin{itemize}
\item {\bf(A1)} There exists an optimal vector field $c$ without singular components and having the form
    $c = \nabla \varphi (\nabla u)$, where $\varphi = \sum_{i=1}^m \varphi_i(t_i)$.
    \item 
{\bf(A2)} All $\varphi_i$ are strictly convex on 
    $[0,1]\setminus\{t: \varphi_i(t)=0\}$.
    \item 
{\bf(A3)} For every $1 \le i \le n$ distribution of $u_{x_i}$ has no atoms on $(0,1]$ (but necessary has an atom at the origin).
\end{itemize}

Note that assumptions (A2)-(A3) are made to exclude "ironing". We know  necessary and sufficient conditions for this for $n=1$. We don't know what are sufficient conditions for this if $n>1$. Note, however, that numerical simulations for $n=m=2$ and $\rho=1$ demonstrate empirical evidence that assumptions (A1)-(A3) are fulfilled in this case (see pictures from \cite{KSTZ}).

 Since $\varphi'_i$ is strictly increasing, then 
$$\max_{1 \le j \le m} c_{i,j}(x_j) = 
\max_{1 \le j \le m} \varphi_i'(u_{x_{i,j}}(x_j))
= \varphi_i'( \max_{1 \le j \le m} u_{x_{i,j}}(x_j))
$$
and  we obtain the following characterization of the optimal mechanism.

\begin{proposition} 
\label{mi2ndprice}
    Let $u, (\varphi_i)$ be a solution to the reduced auctioneer's problem and $c$ be a solution to the corresponding dual problem. Under assumptions (A1)-(A3)  the optimal mechanism $P$ has the following representation
    $$
P_{i,j}(x)  
= I_{ u_{x_{i,j}}(x_j) >  \max_{k \ne j} 
u_{x_{i,k}}(x_k)}
    $$
    for $\overline{\rho}$-almost all $x$.
\end{proposition}

If $n=1$, then the  mechanism described in Theorem \ref{mi2ndprice} is the so-called second-price auction.
Proposition \ref{mi2ndprice} provides a natural generalization of this result: under assumption (A1)-(A3)  the mechanism does not randomize 
and simply chooses the bidder with the largest preference (the largest $u_{x_{i,k}}$).

\begin{remark}\label{n1beck}{\bf  (Case $n=1$ revisited, second price auction)}

Let $n=1$ and assume that the model is regular (see Subsection \ref{cfs1}).
Note that the maximum of
$(u'(x_1), \cdots, u'(x_m))$ is reached at the point
$
\max_{1 \le j \le m} x_j
$
and
$$
P_{j}(x_1, \cdots,x_m) = I_{x_j > \max_{k \ne j} x_k}.
$$
In particular, $P$ has representation
$$
P = \nabla \Phi(x_1, \cdots, x_k),
$$
where
$
\Phi(x) = \max_{1 \le j \le m} (x_m - x_0), 
$
and $[0,x_0] = \{ x: u(x)=0\}$.

To construct $T$ it is sufficient to take any $T_i$ satisfying
$\overline{T}_i(x_i) = \langle \overline{P}_i(x_i), x_i \rangle - u(x_i)$. However, one can prove by direct computation that the { second price auction} given by $P_j$ and
$$
T_j(x) =  
\max_{k \ne j} \max(x_k,x_0)
\cdot I_{x_j > \max_{k \ne j} x_k} 
$$
is optimal and not only ex-ante, but ex-post.
\end{remark}


\section{Open problems and perspectives}

\subsection{Open problems}
\label{opp}
The duality  results for many bidders and items raise many questions. Let us briefly discuss some of them.

\begin{enumerate}

    \item {\bf Can the dual auctioneer's problem for many bidders be interpreted  as an optimal (weak) transportation problem?}

    We will discuss this question in the last section of the paper. Let us also note that any Beckmann's problem is always (\cite{Santambrogio}, \cite{BP}) 
    equivalent to an appropriate congested transport problem, which is a transportation problem on the infinite-dimensional space of curves (and it is a finite-dimensional transshipment problem (\ref{OTP}) only for the case when $\varphi_i$ are indicators of a segment, but this is not the case for $m>1$). Thus, converting the Beckmann's problem into the congested transport problem, we get automatically a kind of infinite dimensional optimal transport problem. But it is not clear, whether this interpretation has some applications.  

\item{\bf Does there always exist a  measure $\pi$ with finite variation giving minimum to the dual monopolist's problem in Theorem \ref{dual=beck}? }

It was mistakenly claimed in \cite{TBK} that such a measure always exists, but the proof contains a gap (the separating functional obtained as a limit can happen to be the zero functional). In fact, this does not seem to be true for the general monopolist problem. But for the auctioneer's problem this might be true. In particular, for $m=1$ bidder the majorizing measure $\pi$ has the same variation as $\mu$ by Theorem \ref{DDT1bid}.

The measure $\pi$ always exists at least for sufficiently regular $\varphi$ and $\rho$, since it has 
representation $\pi = - {\rm div}(\nabla \varphi(\nabla u) \cdot \rho)$ and this always makes sense, because $\varphi$ and $u$ are convex. However, the variation of $\pi$ can be a priori infinite.
In the other hand, we  remark that $\pi$ has always finite $W^1_{l^{\infty}}$-norm for any $\varphi$, which is finite on $X$ and infinite outside of $X$. Indeed, taking a $1$-Lipschits (with respect to $l^{\infty}$-norm) function $v$ one gets
$$
\int_X v d\pi = \int_X \langle \nabla v, c \rangle \rho dx
\le \int_X |c|_{l^1} \rho dx.
$$
But the right-hand side is finite, because $\int_X \varphi^*(c) \rho dx < \infty$. Thus 
we get that the supremum of left-hand side is finite and by duality $\|\pi\|_{W^1_{l^{\infty}}} < \infty$.
This, however, does not mean that $\pi$ has finite variation.
Note than under certain assumptions on $\varphi$, for instance
$\varphi \ge \ f(|\nabla \varphi|) $
 for some increasing $f \ge 0$, we can get a pointwise bound for $\nabla \varphi$ from Proposition 
\ref{ae-exist-lip}.

    \item {\bf Regularity questions}
    
Does there always exist a regular vector field $c$ giving minimum in (\ref{eq_vector_fields_min})? Regularity means that $\varsigma_i=0$. We do not know the answer.
Numerical experiments  from \cite{KSTZ}
demonstrates that $u$ does not look to be continuously differentiable, but $c$ and $\varphi_i$ seem to be smooth. 
It would be interesting to obtain any result on regularity of solutions in general case. 
\end{enumerate}

\subsection{Perspectives: transshipment problem}

It is a natural question whether multi-bidder auction has a relation to transshipment  problem (\ref{OTP}) as the $1$-bidder auction.
In this subsection we get an affirmative answer. However, the relation in this case is more complicated, because, as the reader will see, for $m>1$ it requires additional optimization of the cost function. 

\begin{lemma}
Let $F$ be a bounded vector field.
Set:
$$
d_F(x,y) = \sup_{v : |\nabla v + F|_{\overline{X}^\circ \le 1}} \bigl( v(x) - v(y) \bigr), \ x,y \in \overline{X}.
$$
Then for any differentiable function $v$ the following properties are equivalent:
\begin{enumerate}
\item 
$|\nabla v + F|_{\overline{X}^{\circ}} \le 1$
    \item 
For all $x,y$ one has:
\begin{equation}
    \label{lip2703}
v(x) - v(y) \le d_F(x,y).
\end{equation} 
Equivalently $
-d_F(y,x) \le v(x)-v(y) \le d_F(x,y).
     $
\end{enumerate}
\end{lemma}
\begin{proof}
   Clearly 1) $\Longrightarrow$ 2).
    Assume that  $v$ satisfies (\ref{lip2703}). 
Then any differentiable function $w$ with property 
    $|\nabla w(x) + F(x)|_{\overline{X}^{\circ}} \le 1$
 satisfies the obvious inequality:
  $$
|w(y) - w(x) + \langle F(x), y-x \rangle| \le 
|y-x|_{\overline{X}} + o(|y-x|)
    $$
    for all $y$ from a small neighborhood of $x$.
 In particular, we get
      \begin{equation}
    \label{vinc}
    d_F(y,x) + \langle F(x), y-x \rangle \le 
|y-x|_{\overline{X}} + o(|y-x|).
    \end{equation}
    
    Assume that at some point  $x_0$ one has $|\nabla v(x_0) + F(x_0)|_{\overline{X}^{\circ}} > 1$.
    This means that for a vector  $e$ with $|e|_{\overline{X}}=1$, one has:
    $\langle e, \nabla v(x_0) + F(x_0)\rangle >1$.
    In particular, for small values of $t$
    $$
    v(x_0+te) - v(x_0) + \langle F(x_0), te \rangle =
    \langle te, \nabla v(x_0) + F(x_0) \rangle + o(t) > t + o(t) = |te|_{\overline{X}} + o(t).
    $$
    Thus one has by 2)
    $$
    d_F(x_0+te,x_0) 
     + \langle F(x_0), te \rangle
     \ge 
    v(x_0+te) - v(x_0) + \langle F(x_0), te \rangle  > |te|_{\overline{X}} + o(t).
    $$
    Consequently, inequality (\ref{vinc}) is false for $x=x_0$, $y= x_0+te$ and sufficiently small $t$. We get a contradiction.
 \end{proof}

Using this lemma we relate the Beckmann's problem with constraints to a version of the transshipment problem (\ref{OTP}).
Let us remind the reader the  duality for (\ref{OTP}):
 $$
\sup_{v: v(x)-v(y) \le d_F(x,y)} \int_{\overline{X}} v d \overline{\pi}
= 
 \inf_{\gamma: {\rm Pr}_1 \gamma - {\rm Pr}_2 \gamma = \overline{\pi}} \int_{\overline{X} \times \overline{X}} d_F(x,y) d \gamma.
$$
\begin{remark}
    If $F$ is symmetric, the without loss of generality we can consider only symmetric functions $v$ in the left-hand side, because $d_F$ and $\overline{\pi}$ are symmetric.
\end{remark}

To avoid technicalities, let us consider the $L^2(\overline{\rho}dx)$-dense subspace 
$$
L_b \subset L
$$
of bounded symmetric vector fields from $L$.

Applying the above duality we get the following
\begin{align*}
& \sup_{v \in Lip(| \cdot|_{\overline{X}}, L)} \int_{\overline{X}} v d \overline{{\pi}} =
\sup_{v, F \in L: |v + \nabla F|_{\overline{X}^{\circ}} \le 1} \int_{\overline{X}} v d \overline{{\pi}} =  \sup_{v, F \in L_b: |v + \nabla F|_{\overline{X}^{\circ}} \le 1} \int_{\overline{X}} v d \overline{{\pi}} \\& = 
\sup_{F \in L_b} \Bigl( \sup_{v: v(x)-v(y) \le d_F(x,y)} \int_{\overline{X}} v d \overline{\pi}\Bigr)
= \sup_{F \in L_b} \Bigl(
 \inf_{\gamma: {\rm Pr}_1 \gamma - {\rm Pr}_2 \gamma = \overline{\pi}} \int_{\overline{X} \times \overline{X}} d_F(x,y) d \gamma \Bigr).
 \end{align*}
 We obtain that the dual Beckmann's problem with linear constraints is related to the transshipment problem in the following way
\begin{theorem} {\bf (Relation to the optimal transshipment problem)}
\label{trpr}
$$\inf_{\overline{c} \in L^{\bot}, - {\rm div}_{{\rho}} c = {\pi} } \int_{\overline{X}} |\overline{c}|_{\overline{X}} \overline{\rho} dx
= \sup_{F \in L_b} \Bigl(
 \inf_{\gamma: {\rm Pr}_1 \gamma - {\rm Pr}_2 \gamma = \overline{\pi}} \int_{\overline{X} \times \overline{X}} d_F(x,y) d \gamma \Bigr).
 $$
\end{theorem}

\begin{remark} 
{\bf (Auction design and optimal transshipment problem (\ref{OTP})).}
We obtain the following corollary for the auctioneer's problem.
    The value of the auction design problem equals
    $$
     \inf_{\pi \succeq \mu} \sup_{F \in L_b}  \Bigl(
 \inf_{\gamma: {\rm Pr}_1 \gamma - {\rm Pr}_2 \gamma = \overline{\pi}} \int_{\overline{X} \times \overline{X}} d_F(x,y) d \gamma \Bigr).
 $$
 Here
 $
d_F(x,y) = \sup_{v : |\nabla v + F|_{\overline{X}^\circ \le 1}} \bigl( v(x) - v(y) \bigr),
$
where the supremum is taken over symmetric differentiable functions $v$.

We hope to justify the following heuristic observations in the subsequent work.
\begin{itemize}
    \item
 Formally applying minimax principle we can rewrite the functional in a bit more elegant form
 $$
     \sup_{F \in L_b}  \Bigl(
 \inf_{\gamma: {\rm Pr}_1 \gamma - {\rm Pr}_2 \gamma = \overline{\pi}, \pi \succeq \mu} \int_{\overline{X} \times \overline{X}} d_F(x,y) d \gamma \Bigr).
 $$
\item
    Remarkably,  the case $n=1$ can be also reduced to the  weak transportation problem without additional optimization on $F$. Here below is the heuristic argument. According to Remark \ref{n1beck} the optimal mechanism has the gradient representation $P = \nabla \Phi$ in the regular case. This means that we can take $F=0$ in the optimal mechanism. Hence the value of the auctioneer's problem is
$$
\inf_{\gamma: {\rm Pr}_1 \gamma - {\rm Pr}_2 \gamma = \overline{\pi}, \pi \succeq \mu} \int_{\overline{X} \times \overline{X}} |x-y|_{\overline{X}} d \gamma.
$$
\end{itemize}
\end{remark}

\end{document}